\documentclass[nonacm,format=acmsmall, review=false, screen=true,]{acmart}

\setcopyright{none}
\renewcommand\footnotetextcopyrightpermission[1]{}
\acmConference{}{}{}
\acmBooktitle{}
\acmDOI{}
\acmISBN{}

\usepackage{amsmath}
\usepackage{amsfonts}
\usepackage{booktabs}

\def\ve#1{\mathchoice{\mbox{\boldmath$\displaystyle\bf#1$}}
{\mbox{\boldmath$\textstyle\bf#1$}}
{\mbox{\boldmath$\scriptstyle\bf#1$}}
{\mbox{\boldmath$\scriptscriptstyle\bf#1$}}}
\newcommand\vea{{\ve a}}

\newcommand\veb{{\ve b}}
\newcommand\vecc{{\ve c}}
\newcommand\ved{{\ve d}}
\newcommand\vece{{\ve e}}

\newcommand\veh{{\ve h}}

\newcommand\vep{{\ve p}}

\newcommand\ver{{\ve r}}
\newcommand\ves{{\ve s}}

\newcommand\vev{{\ve v}}
\newcommand\vew{{\ve w}}
\newcommand\vex{{\ve x}}
\newcommand\vey{{\ve y}}

\newcommand\vezeta{{\ve \zeta}}

\newcommand\veeps{{\ve \varepsilon}}

\newcommand\veDelta{{\ve \Delta}}
\newcommand\vemu{{\ve \mu}}
\newcommand\venu{{\ve \nu}}
\newcommand\vezero{{\ve 0}}
\newcommand\veone{{\ve 1}}
\newcommand{\R}{\mathbb{R}}
\newcommand{\Z}{\mathbb{Z}}
\newcommand{\N}{\mathbb{N}}
\newcommand{\Oh}{\mathcal{O}}
\DeclareMathOperator{\poly}{\mathrm{poly}}

\newcommand{\pref}{\ensuremath{\succ}}

\newcommand{\niceproblem}[3]{
    \begin{probbox}{#1}
        \textbf{Input:} #2 \\
        \textbf{Output:} #3
    \end{probbox}
}
\newenvironment{probbox}[1]{%
    \begin{center}
    \begin{minipage}{0.95\linewidth}
    \hrule
    \smallskip
    \textbf{#1}\smallskip\\
}{%
    \smallskip
    \hrule
    \end{minipage}
    \end{center}
}

\newenvironment{proofclaim}{\noindent{\em Proof of the claim.}}{\qedclaim}
\newcommand{\qedclaim}{\hfill $\diamond$ \medskip}

\DeclareMathOperator{\rank}{rank}
\newcommand{\CandRank}[2]{
  \rank_{#1}(#2)
}

\newcounter{lp}

\newtheorem{theorem}{Theorem}
\newtheorem{lemma}{Lemma}
\newtheorem{corollary}{Corollary}
\newtheorem{claim}{Claim}
\newtheorem{proposition}{Proposition}
\newtheorem{openproblem}{Open Problem}
\theoremstyle{remark}
\newtheorem{remark}{Remark}

\title{Continuous Computational Social Choice: A Case Study in Bribery}

\author{Martin Koutecký}
\affiliation{%
  \institution{CSI, Faculty of Mathematics and Physics, Charles University}
  \city{Prague}
  \country{Czech Republic}
}
\email{koutecky@iuuk.mff.cuni.cz}
\author{Nikolaos Melissinos}
\affiliation{%
  \institution{CSI, Faculty of Mathematics and Physics, Charles University}
  \city{Prague}
  \country{Czech Republic}
}
\email{melissinos@iuuk.mff.cuni.cz}
\author{Tung Anh Vu}
\affiliation{%
  \institution{CSI, Faculty of Mathematics and Physics, Charles University}
  \city{Prague}
  \country{Czech Republic}
}
\email{tung@iuuk.mff.cuni.cz}
\author{Lluís Sabater}
\affiliation{%
  \institution{CSI, Faculty of Mathematics and Physics, Charles University}
  \city{Prague}
  \country{Czech Republic}
}
\email{llsabater@iuuk.mff.cuni.cz}

\begin{document}

\begin{abstract}
Computational social choice seeks algorithmic answers to questions about preference aggregation, safety of elections, robustness of outcomes, stability, etc. It overwhelmingly models societies as composed of \emph{discrete} agents.

We propose to study computational social choice problems in a \emph{society continuum} setting, where a society is modeled as a distribution of infinitely many infinitesimal agents of different types.
An analogous approach has been very useful in physics (it is the basis of statistical mechanics), economics (mean field games), and other fields.

As an initial case study, we focus on \emph{election attacks} (bribery and control), which have been extensively studied in the discrete setting.
We show that a broad class of standard election attacks becomes polynomial-time solvable in the society continuum. The class contains problems that are NP-hard discretely, among them Borda- and Bucklin-\textsc{CCDV} and unit-cost Borda-\textsc{Swap Bribery}.
Furthermore, we give polynomial-time algorithms for $k$-Approval-\textsc{Swap Bribery} when $k$ is constant for general costs, and when $k$ varies and the cost function is \emph{additively separable}.
The latter result contrasts with the discrete problem, which we prove NP-complete for additively separable costs and every fixed $k\ge 2$.
In contrast, we prove that Borda-\textsc{Swap Bribery} and $k$-Approval-\textsc{Swap Bribery}, both with general costs, remain computationally hard in the society continuum.

To obtain these results, we use both continuous and discrete optimization techniques, such as the Configuration LP framework and dynamic programming.
Of particular note is the technique underlying our hardness proofs, which shows how to ``reverse the flow of hardness'' between LP formulations and pricing problems.
\end{abstract}

\maketitle

\section{Introduction}
\label{sec:intro}

Modern real-world election campaigning models a society as composed of voter segments which both vote the same and are amenable to be influenced by the same methods, at the same cost~\cite{GreenGerber2019GOTV,Issenberg2012VictoryLab}.
The natural phrasings work in percentage points:
``shift about $1\%$ of segment~$t$ from preferring $a$ to preferring $b$,'' or
``increase turnout of type~$t$ by $2\%$,'' subject to a budget and a cost.
The political science research on these topics is reported in the same terms~\cite{GreenGerber2019GOTV}.
A common observation after Brexit was that ``If [...] Remain [...] managed to persuade \emph{two in a hundred} more voters [...] it would have won.''~\cite{Cassidy2016Remain}
In short, it is common and sensible to think, talk, and write about a society as a continuum of individuals distributed over finitely many types.

This perspective change from discrete to continuous has been fruitful in many fields.
It is standard in statistical mechanics, where mean-field arguments track the fraction of the population in each state rather than individual particles~\cite{pathria2016statistical}; related distributional models appear in sociophysics~\cite{castellano2009statistical,galam2012sociophysics}, in quantitative studies of social influence~\cite{pentland2014social,acemoglu2011opinion}, in mean-field limits of large agent systems~\cite{gomes2010discrete}, and in ``geometry of voting'' approaches that replace finite electorates by distributions over preference structures~\cite{saari1994GeometryOfVoting}.

The continuum limit is not merely a \emph{succinct encoding} of a large
population (akin to high-multiplicity encodings in scheduling~\cite{hochbaumShamir1991HM} or elections~\cite{fitzsimmonsHemaspaandra2019HM}).
It changes the \emph{space of feasible solutions}, and through this, the computational complexity of problems.
In an election attack problem, for example, an attacker can now move \emph{any fraction} of mass between types, rather than being restricted to an integer number of voters.
This can make problems which are hard in the discrete variant tractable in the continuous variant.

This \emph{computational complexity} aspect of continuous modeling has not been studied before in the context of computational social choice, and we initiate its study here.
The lens applies broadly across computational social choice, giving rise to \emph{continuous} computational social choice, a mirror image of the classical discrete view that has been dominant so far.
In this paper we carry out a first case study: the complexity impact of this modeling shift for \emph{election attacks}.

\subsection{Our Contribution}

We defer the formal definitions of elections, voting, and bribery problems to Section~\ref{sec:prelim}; here we just note that the problems we consider are natural analogues of classical discrete problems where a society of $n$ discrete voters is replaced by a continuum of voters distributed over $\tau \in \N$ types.
This can be viewed as the limit of taking increasingly larger election instances with fixed ratios
of voters of each type, until the total population becomes the continuum; hence, we dub it the \emph{society continuum} model.
We use an $\infty$ subscript to denote the continuous variant of a discrete problem, e.g., we write ``Borda-\textsc{Swap Bribery}$_\infty$'' for the continuous analogue of the classical Borda-\textsc{Swap Bribery} problem.

\paragraph{Techniques}
More than just individual positive and negative results, our approach provides a general template and is deeply connected to the techniques underneath it, so we begin with them.
The key technical tool in all of our results is the \emph{Configuration LP}.
The natural linear programming formulation of many election attack problems has exponentially many variables but only polynomially many constraints.
The Configuration LP approach then considers the dual LP, which has polynomially many variables and exponentially many constraints, and solves it by the ellipsoid method with a separation oracle.
Equivalently, the primal LP is solved by column generation, where the dual separation problem becomes the \emph{pricing problem}: finding a variable of negative reduced cost.
The general template for treating our problems thus becomes to define, for the specific voting rule and attack, the primal LP, its pricing problem, and then examine its complexity.

If the pricing problem is tractable, as we show in several cases, this immediately implies tractability of the original voting attack problem.
However, the converse (hardness of pricing implying hardness of attack) does not hold in general; the easy direction is only that hardness of attack implies hardness of pricing.
In this case, our contribution is in showing what kind of arguments make this hardness transfer from pricing to the attack problem possible.
To the best of our knowledge, this technique has not been used in computational social choice before.

\paragraph{Tractable cases}
We first identify a broad class of voting rules and attack problems which become polynomial-time solvable in the society continuum (Theorem~\ref{thm:poly-lp}).
The class contains, for example, every scoring protocol, Condorcet's rule, and Bucklin's rule, under attacks such as \textsc{Shift Bribery} and Constructive Control by Deleting Voters (CCDV), and combinations thereof.
The technique behind this result is simple: each such problem can be translated into polynomially many polynomial-size linear programs.

The attacks \textsc{Bribery}, its priced variant \textsc{\$Bribery}, and \textsc{Swap Bribery} are not covered by this class.
In these problems the attacker can change a voter's preference order arbitrarily, so the set of rankings appearing in an optimal solution is no longer determined by the original types and polynomially bounded in size.
Still, for every scoring protocol we show that \textsc{\$Bribery}$_\infty$ is polynomial-time solvable, because its pricing problem reduces to simple sorting (Theorem~\ref{thm:score-dollar-bribery}).

We push beyond this boundary in two further directions.
For Borda, we show that \textsc{Shift Bribery}$_\infty$ and unit-cost \textsc{Swap Bribery}$_\infty$ are polynomial-time solvable (Corollary~\ref{cor:spwr-applications}); discretely, Borda-\textsc{Shift Bribery} is NP-hard~\cite{elkind2009swapbribery}, even with unit costs~\cite{DBLP:journals/iandc/BredereckCFNN16}.
For $k$-Approval-\textsc{Swap Bribery}$_\infty$, we give an XP algorithm parameterized by $k$ (Corollary~\ref{thm:XP:algo:k-appr}).
Its running time is $m^{\Oh(k)}\poly(L)$, so it is polynomial for every constant $k$.
We also give a polynomial-time algorithm for arbitrary $k$ when the cost function is \emph{additively separable} (Corollary~\ref{cor:k-approval-additive-poly}).
Additive separability is a natural and broad extension of unit costs: each position $j$ of a voter's ranking carries a demotion price $a_j$ and a promotion price $b_j$, and swapping the candidates at positions $j < j'$ costs $a_j + b_{j'}$; unit costs are the special case $a_j = 0$, $b_j = 1$.
What a voter charges thus depends on \emph{where} in her ranking she is asked to change her mind, not on \emph{which} candidates are involved -- the natural regime whenever persuasion cost is driven by salience rather than by the identity of the candidates.
For example, additively separable costs can model that displacing a voter's favorite is expensive while reshuffling the tail is cheap.
It also covers the case in which each position carries an ``esteem'' value $\alpha_1 \ge \cdots \ge \alpha_m$ and a swap is paid for by the esteem it destroys plus a flat fee $\beta$ per swap performed, i.e., $\alpha_j - \alpha_{j'} + \beta$.
In contrast, we prove that discrete $k$-Approval-\textsc{Swap Bribery} with additively separable costs is NP-complete for every fixed $k\ge 2$ (Theorem~\ref{thm:discrete-kapproval-separable}).

\paragraph{Hardness}
Tractability does not extend to general costs: we prove that Borda-\textsc{Swap Bribery}$_\infty$ is NP-hard, even with a single \emph{bribable} voter type (Corollary~\ref{cor:margin-to-borda}), and that no algorithm solves $k$-Approval-\textsc{Swap Bribery}$_\infty$ in time $f(k)\poly(L)$ unless $W[1]=FPT$ (Theorem~\ref{k-approval:hardness}).
So the XP algorithm above is unlikely to be improvable to FPT.

Our results in fact concern the \emph{margin} version of these problems, in which $c^*$ must beat each opponent $c$ by a prescribed margin $\Delta_c$ rather than merely tie it. This is a natural generalization -- a campaign manager rarely wants a bare win -- and it is the main component which makes our hardness transfer work.
With general costs, the margin version actually collapses back to the original problem: we show how to encode the margin requirement using gadgets composed of unbribable voter types (Corollary~\ref{cor:margin-to-borda}).

Remarkably, a slight generalization of unit cost Borda-\textsc{Swap Bribery}$_\infty$ to \emph{uniform} costs, where each type $i$ has a single per-swap cost $q_i$, turns out to be a natural yet apparently so far unstudied optimization problem about permutations.
We state it as Open Problem~\ref{op:uniform} and consider it an important open problem of independent interest.

\newcommand{\tocont}{\,\raisebox{0.1ex}{$\scriptstyle\triangleright$}\,}
\newcommand{\cont}[1]{\textbf{#1}}
\newcommand{\PNP}{\ensuremath{\mathrm{P}^{\mathrm{NP}}_{\|}}}

\begin{table*}[t]
\centering
\small
\setlength{\tabcolsep}{3.5pt}
\resizebox{\textwidth}{!}{%
\begin{tabular}{@{}l|ccc|ccc|c@{}}
\toprule
 & & & & \multicolumn{3}{c|}{\textsc{Swap Bribery}, by cost model} & \\
\cmidrule(lr){5-7}
rule & \textsc{\$Bribery} & \textsc{Shift Bribery} & \textsc{CCDV} & unit & add.\ separable & general & winner \\
\midrule
scoring protocols   & NP-h$^{a}$\tocont\cont{P}$^{1}$ & NP-h$^{b}$\tocont\cont{P}$^{2}$ & NP-h$^{a}$\tocont\cont{P}$^{2}$ & NP-h$^{c}$\tocont ? & NP-h$^{c}$\tocont ? & NP-c$^{t}$\tocont\cont{NP-h}$^{3}$ & P\tocont P \\
\quad Borda         & \multicolumn{1}{c}{$\uparrow$} & \multicolumn{1}{c}{$\uparrow$} & \multicolumn{1}{c|}{$\uparrow$} & NP-c$^{c}$\tocont\cont{P}$^{2}$ & NP-c$^{c}$\tocont ? & \multicolumn{1}{c|}{$\uparrow$} & P\tocont P \\
\quad $k$-Approval  & \multicolumn{1}{c}{$\uparrow$} & P$^{b}$\tocont\cont{P}$^{2}$ & \multicolumn{1}{c|}{$\uparrow$} & P$^{d}$\tocont\cont{P}$^{2}$ & NP-c$^{8}$\tocont\cont{P}$^{4}$ & NP-c$^{d}$\tocont\cont{XP$_k$}$^{5}$, \cont{no FPT}$^{6}$ & P\tocont P \\
Condorcet           & NP-c$^{e}$\tocont\cont{P}$^{2}$ & NP-c$^{h}$\tocont\cont{P}$^{2}$ & NP-c$^{f}$\tocont\cont{P}$^{2,g}$ & NP-c$^{h}$\tocont\cont{P}$^{2,g}$ & NP-c$^{h}$\tocont ? & NP-c$^{h}$\tocont ? & P\tocont P \\
Bucklin             & NP-c$^{l}$\tocont ? & P$^{i}$\tocont\cont{P}$^{2}$ & NP-c$^{j}$\tocont\cont{P}$^{2}$ & ?\tocont ? & ?\tocont ? & NP-c$^{l}$\tocont ? & P\tocont P \\
Dodgson, Young      & \multicolumn{6}{c|}{\PNP-h$^{q}$\tocont ?} & \PNP-c$^{o}$\tocont P$^{g}$ \\
Kemeny, Slater      & \multicolumn{6}{c|}{\PNP-h$^{p,q}$\tocont\PNP-h$^{7}$} & \PNP-c$^{k}$\tocont\PNP-c$^{7}$ \\
\bottomrule
\end{tabular}%
}
\caption{The bribery landscape, discrete\tocont continuum. Each cell gives the complexity of the
discrete problem, then of its society-continuum analogue; \cont{bold} marks a continuum entry
proved here. \emph{Superscript numbers point to our results, letters to the literature.}
$\uparrow$: covered by the row above. \,?\,: open, to our knowledge.
$\PNP$: parallel access to NP, a class containing NP.\\
\textbf{Ours:} $1$~Thm.~\ref{thm:score-dollar-bribery}; $2$~Cor.~\ref{cor:spwr-applications};
$3$~Cor.~\ref{cor:margin-to-borda}; $4$~Cor.~\ref{cor:k-approval-additive-poly};
$5$~Cor.~\ref{thm:XP:algo:k-appr}; $6$~Thm.~\ref{k-approval:hardness};
$7$~homogeneity (Section~\ref{sec:related}); $8$~Thm.~\ref{thm:discrete-kapproval-separable}.\\
\textbf{Prior:} $a$~\citealp{hemaspaandraSchnoor2016dichotomy} ($k \ge 3$);
$b$~\citealp{elkind2009swapbribery}; $c$~\citealp{BaumeisterHR19};
$d$~\citealp{dornSchlotter2012multivariate};
$e$~via Voter Replacement~\citealp{elkindFaliszewskiSlinko2012HammingDR};
$f$~\citealp[Prop.~1]{faliszewskiKarpovObraztsova2022groupSeparable};
$g$~homogeneous variant already polynomial~\citealp{young1977ExtendingCondorcet,fishburn1977Condorcet,rotheSpakowskiVogel2003YoungExactComplexity};
$h$~\citealp{bartholdiToveyTrick1989Dodgson,DBLP:journals/jacm/HemaspaandraHR97};
$i$~\citealp{schlotterFaliszewskiElkind2017campaign};
$j$~\citealp{erdelyiFellowsRotheSchend2015bucklin};
$k$~\citealp{DBLP:journals/tcs/HemaspaandraSV05,DBLP:conf/stacs/Lampis22};
$l$~\citealp[Thms.~4.1 and~5.4]{faliszewskiReischRotheSchend2015bucklin};
$o$~\citealp{DBLP:journals/jacm/HemaspaandraHR97,rotheSpakowskiVogel2003YoungExactComplexity};
$p$~\citealp[Thm.~4.8]{brandtBrillHemaspaandraHemaspaandra2015bypassing};
$q$~at budget~$0$ each attack \emph{is} winner determination, and Young-\textsc{CCDV} is in fact
$\Sigma^p_2$-complete~\citealp[Thm.~9]{fitzsimmonsHemaspaandraHooverNarvaez2019control};
$t$~except plurality and veto~\citealp{elkind2009swapbribery,betzlerDorn2010possibleWinner,baumeisterRothe2012finalStep}.}
\label{tab:landscape}
\end{table*}

\subsection{Related Work}
\label{sec:related}

\paragraph{Geometry of voting}
Saari has developed a ``geometry of voting'' program~\cite{saari1994GeometryOfVoting,saari1995BasicGeometry,saari2001DecisionsElections}, in which he views the simplex as the space of possible election profiles, and views voting rules as inducing a partition of this space into regions won by each candidate.
For this to be well-defined, the voting rule needs to be \emph{homogeneous}: cloning each voter any number of times should not change the outcome.
Most rules, e.g.\ scoring protocols and Condorcet's, are homogeneous; Dodgson's and Young's are not, see below.
Saari's work is \emph{explanatory and axiomatic} -- continuous representations diagnose paradoxes and compare voting rules -- and is not concerned with computational complexity.

\paragraph{Homogeneity and ``homogenized'' Condorcet extensions}
The lack of homogeneity has been viewed as a defect for several Condorcet-style
rules.
Already \citet{young1977ExtendingCondorcet} remarks that the possibility
that cloning an electorate can change the winner is an ``absurdity''~\cite{mccabeDansted2008DodgsonAbsurdity}.
This triggered a line of work that studies ``homogeneous'' variants of the
``closest-to-Condorcet'' rules, in particular Fishburn's homogeneous variant of
Dodgson's rule (often called \emph{Dodgson Clone})~\cite{fishburn1977Condorcet}.
This ``homogeneous Dodgson'' is what we would call ``Dodgson in the society continuum''.
\citet{rotheSpakowskiVogel2003YoungExactComplexity} show that winner and ranking for Young's rule are
$\mathrm{P}^{\mathrm{NP}}_{||}$-complete, but that homogeneous
Dodgson is polynomial-time solvable via
relaxing the ILP formulation for \textsc{Dodgson Score} of
\citet{bartholdiToveyTrick1989Dodgson}.

The other two well-known $\mathrm{P}^{\mathrm{NP}}_{||}$-complete voting rules are Kemeny~\cite{DBLP:journals/tcs/HemaspaandraSV05} and Slater~\cite{DBLP:conf/stacs/Lampis22}.
Both are homogeneous, and thus remain hard even in the society continuum.

\paragraph{Attack problems}
Within the distance-rationalizability framework~\cite{meskanenNurmi2008DistRational}, \citet{elkindFaliszewskiSlinko2012HammingDR} isolate a ``voter replacement rule'' (VRR): the winner minimizes the number of voters that must be replaced by arbitrary votes so that the candidate becomes a Condorcet winner.
Young's, Dodgson's, and VR rules are each defined by the minimum cost of an attack which makes a candidate a Condorcet winner: the minimum number of deleted voters (Condorcet-\textsc{CCDV}), of adjacent swaps (Condorcet-\textsc{Swap Bribery}), and of fully replaced voters (Condorcet-\textsc{Bribery}), respectively.
In the society continuum, these three attacks fall under Theorem~\ref{thm:poly-lp}.

\paragraph{Homogeneous attack models}
In another sense, we study \emph{homogenized} attack
problems.
These have a different motivation from the one for homogenizing Dodgson/Young, because (non-homogeneous) bribery is not ``absurd'' in the same sense.
Indeed, cloning every voter $q$ times and cloning the attack with them again makes $c^*$ win.
What can change under cloning is the \emph{minimum cost} of achieving the goal:
in the ``up-scaled'' election (with cloned voters), there may exist a cheaper attack that exploits the larger electorate.

The society continuum captures the large-scale (or mean-field) limit of this phenomenon.
Fix a continuous type distribution $\vemu$ and let $\mathrm{OPT}_q(\vemu)$ be the minimum cost of an integer attack in an electorate of size $q$ realizing $\vemu$; for large electorates the relevant quantity is the asymptotics of $\mathrm{OPT}_q(\vemu)/q$.
Studying this limit removes rounding artefacts and can, for example, show which hardness results are driven by the discreteness of voters and which stem from a different source, e.g., the discreteness of the set of candidates.

\paragraph{Other fractional objects in computational social choice}
Continuity has also entered computational social choice from several directions.
\citet{xiaConitzerProcaccia2010scheduling} let each manipulator cast a convex combination of rankings and obtain a polynomial-time algorithm for coalitional manipulation under scoring rules; their electorate stays discrete, and the divisibility of the non-manipulators' votes is immaterial to the results.
\citet{freemanBrillConitzer2015tiebreaking} extend homogeneous voting rules from integer profiles to all of $\R_{\ge 0}^{m!}$ to study tie-breaking.
\citet{meir2015plurality} studies plurality dynamics in a model with no finite set of voters, where a real fraction of each type is assigned to each vote; he asks about equilibrium and convergence.
In each case one ingredient of the setting becomes continuous while the others stay discrete.

\paragraph{Separation and optimization as a source of hardness}
Using the separation/optimization equivalence ``in the negative'' goes back to \citet{GLS1981}, who used it to show hardness of computing the weighted fractional chromatic number.
Since then, we are only aware of three further examples: the hardness of fractional Steiner tree packing~\cite{JainMahdianSalavatipour2003packing}, of equilibrium computation in security games~\cite{Xu2016security}, and of signaling in Bayesian zero-sum games~\cite{bhaskarChengKoSwamy2016signaling}.

\section{Preliminaries}
\label{sec:prelim}

With small exceptions, we follow the standard definitions and terminology, see e.g. the textbook~\cite{HandbookComSoC}.

\paragraph{Elections}

An \emph{election~$E=(C,V)$} consists of a set $C$ of \emph{candidates} and a set~$V$ of \emph{voters}, who indicate their preferences over the candidates in $C$.
Usually $n=|V|$ and $m=|C|$.
We use the ordinal model: each voter $v$'s preferences are a \emph{preference order} $\pref_v$, a total order over $C$, or $\emptyset$ if the voter is inactive.
Let $\mathcal{R}$ be the set of rankings (linear orders) of $C$.
For any $\rho\in \mathcal{R}$ and $c\in C$, $\textrm{rank}(c,\rho)$ denotes the rank of the candidate $c$ in the ranking $\rho$; the most preferred candidate has rank 1 and the least preferred one rank $|C|$; equivalently, $c$ sits in \emph{position} $\textrm{rank}(c,\rho)$ of $\rho$.
For distinct candidates~$c,c'\in C$, we write $c\pref_\rho c'$ if $\rho$ ranks~$c$ before~$c'$.
When displaying a ranking as a chain, we put the most preferred candidate on the left.

A \emph{voting rule} is a function which maps an election $E$ to a set of winners $W \subseteq C$.
A candidate $c$ is a \emph{Condorcet winner} if for every other candidate $c'$, the number of voters who prefer $c$ to $c'$ is greater than the number of voters who prefer $c'$ to $c$; a \emph{weak Condorcet winner} (or \emph{co-winner}) is defined with ``at least'' in place of ``greater than''.
Throughout the paper we use the \emph{co-winner} convention: making $c^*$ win means making it a (co-)winner under non-strict inequalities, for Condorcet as well as for scoring protocols.
The convention costs little: For the linear winner conditions used here, the strict-winner variant can be decided by introducing a common slack in all winning inequalities, maximizing it subject to the budget constraint, and testing whether its optimum is positive.
A \emph{scoring protocol} is defined by a vector $\ves = (s_1, \ldots, s_m)$ of non-negative integers such that $s_1 \geq s_2 \geq \cdots \geq s_m$; under such a protocol, each voter $v$ awards $s_j$ points to $c$ if $\text{rank}(c,\pref_v)=j$.
Examples include Plurality ($\ves = (1,0,\ldots,0)$), $k$-Approval ($\ves = (1,\ldots,1,0,\ldots,0)$ with $k$ ones), and Borda ($(m-1,m-2,\ldots,0)$).
Every voting rule ignores all voters with preference $\emptyset$.

\paragraph{Voter Types and Societies}

We capture election attacks by an abstract notion of a ``move'' of population between voter types, in the spirit of~\citet{DBLP:conf/atal/KnopKM18}.

Two voters are of the same \emph{type} if they have identical preferences and identical bribery costs; a bribery action can only affect a voter's preferences, not their cost function.
Let $\mathcal{R}_{\emptyset} = \mathcal{R} \cup \{\emptyset\}$.
A \emph{move cost function} is a function $\pi: \mathcal{R}_\emptyset \rightarrow \mathbb{R}_{\geq 0} \cup \{+\infty\}$, where $\pi(\rho)$ is the per-unit cost of changing a voter's preference to $\rho$.
Thus, a \emph{voter type} is defined by a pair $(\sigma,\pi)$ with $\sigma \in \mathcal{R}_\emptyset$ its initial ranking and $\pi$ a move cost function.
We require $\pi(\sigma)=0$, so leaving mass at its initial ranking has no cost.
Throughout, $\sigma$ denotes the initial ranking of a type and $\rho$ ranges over destination rankings.

Given a set of $\tau$ voter types $\{(\sigma_i, \pi_i)\}_{i \in [\tau]}$, a \emph{society} is defined by a vector $\vemu = (\mu_1, \ldots, \mu_\tau) \in \mathbb{R}_{\geq 0}^\tau$, where $\mu_i$ indicates the population, or \emph{mass}, of the $i$-th type.

\paragraph{Bribery Actions as Moves}
A \emph{move} is defined by a vector $\vex = (x_{i \to \rho})_{i \in [\tau],\, \rho \in \mathcal{R}_\emptyset}$, where $x_{i \to \rho} \in \mathbb{R}_{\geq 0}$ indicates the amount of population of type $i$ whose preference is changed to $\rho$.
A move is \emph{feasible} for a society $\vemu$ if, for all $i \in [\tau]$, $\sum_{\rho \in \mathcal{R}_\emptyset} x_{i \to \rho} = \mu_i$, that is, $\vex$ can be viewed as a matrix whose row marginals are $\vemu$.

A voting rule is indifferent to move costs, so the \emph{output society} after applying $\vex$ to $\vemu$ is $\vemu' = (\mu'_\rho)_{\rho \in \mathcal{R}_\emptyset}$ with $\mu'_\rho = \sum_{i \in [\tau]} x_{i \to \rho}$.
(Slightly abusing notation, $\vemu'$ is indexed by rankings rather than types.)
When evaluating a voting rule or winning condition, we use the restriction $(\mu'_\rho)_{\rho\in\mathcal{R}}$; the coordinate of $\emptyset$ is ignored.
Given the move costs, let $T = |\{\rho \in \mathcal{R}_\emptyset \mid \exists i: \pi_i(\rho) < +\infty\}|$ be the number of \emph{allowed output rankings}, the destinations reachable at finite cost.
\citet{DBLP:conf/atal/KnopKM18} show how various election attacks can be modeled using moves.
In brief: \textsc{Swap Bribery} charges $\pi_i(\rho) = \sum_{(c,c') \in \text{inv}(\sigma_i, \rho)} \delta_i(c,c')$, where $\text{inv}(\sigma_i, \rho)$ is the set of pairs inverted between $\sigma_i$ and $\rho$ and $\delta_i(c,c') \ge 0$ is the given cost of swapping $c$ and $c'$ in a vote of type $i$. \textsc{\$Bribery} charges a flat $\gamma_i$ per unit of type $i$ bribed; the equally short cost functions of \textsc{Shift Bribery} and \textsc{CCDV} are below.
We classify \textsc{Swap Bribery} instances by their swap costs: costs are \emph{uniform} if each type $i$ has a single per-swap cost $q_i \ge 0$ with $\delta_i(c,c') = q_i$ for all pairs $c \neq c'$, and \emph{unit} if in addition $q_i = 1$ for every type $i$ (so that $\pi_i(\rho) = |\text{inv}(\sigma_i,\rho)|$ is the Kendall tau distance); costs that need not have this structure are \emph{general}.

In each case studied here the move cost function is encoded succinctly -- by $\Oh(m^2)$ numbers or fewer -- and $\pi_i(\rho)$ is computable from that encoding in time $\poly(m)$.

For a scoring rule $\mathcal{X}$, the \emph{margin} variant, written $\mathcal{X}$~Margin-\textsc{Swap Bribery}$_\infty$, additionally specifies target margins $\veDelta \in \mathbb{Q}^{C \setminus \{c^*\}}$.
The move must make the score difference of $c^*$ over every $c\neq c^*$ at least $\Delta_c$.
The original problem is the special case $\veDelta=\vezero$.

\subsection{Formal move cost functions of the standard attacks}\label{app:attack-costs}
\begin{itemize}
	\item \textsc{\$Bribery}: for every type $(\sigma_i, \pi_i)$ and ranking $\rho \neq \sigma_i$, $\pi_i(\sigma_i) = 0$ and $\pi_i(\rho) = \gamma_i$ for some constant $\gamma_i \geq 0$; $\pi_i(\emptyset) = +\infty$.
	\item \textsc{Swap Bribery}: for every type $(\sigma_i, \pi_i)$ and $\rho \in \mathcal{R}$, $\pi_i(\rho)$ is $\sum_{(c,c') \in \text{inv}(\sigma_i, \rho)} \delta_i(c,c')$, where $\text{inv}(\sigma_i, \rho)$ is the set of pairs of candidates $(c,c')$ such that $c \pref_{\sigma_i} c'$ but $c' \pref_{\rho} c$ or vice versa, so $\delta_i(c,c')$ is the cost of swapping $c$ and $c'$ in a vote of type $i$; $\pi_i(\emptyset) = +\infty$. We allow $\delta_i(c,c') \in \mathbb{Q}_{\ge 0} \cup \{+\infty\}$; a pair of cost $+\infty$ makes every ranking inverting it unreachable at finite cost, and a type all of whose swap costs are $+\infty$ is \emph{unbribable}. After the hardness results, we show that this convenience can be removed (Lemma~\ref{lem:remove-infinite-costs}).
	\item \textsc{Shift Bribery}: each type $(\sigma_i, \pi_i)$ is equipped with a nondecreasing cost function $\xi_i \colon \{0, 1, \dots, m-1\} \to \mathbb{Q}_{\geq 0}$ with $\xi_i(0) = 0$; for every $\rho \in \mathcal{R}$ obtained from $\sigma_i$ by shifting the preferred candidate up by $k$ positions, $\pi_i(\rho)=\xi_i(k)$; for every other $\rho$ we set $\pi_i(\rho) = +\infty$, and $\pi_i(\emptyset) = +\infty$. (Unit-cost \textsc{Shift Bribery} is the special case $\xi_i(k)=k$.)
	\item \textsc{Constructive Control by Deleting Voters (CCDV)}: for every type $(\sigma_i, \pi_i)$, $\pi_i(\sigma_i) = 0$ and $\pi_i(\rho) = +\infty$ for every $\rho \in \mathcal{R} \setminus \{\sigma_i\}$; $\pi_i(\emptyset) = \gamma_i$ for some constant $\gamma_i \geq 0$.
\end{itemize}

\begin{remark}
	Our notions of types and moves could easily be extended to capture problems such as CCAV (control by adding voters) and \textsc{Support Bribery} (where a voter reports an approval ballot of their top-$k$ candidates, and a briber can change the threshold $k$); we keep the definitions simpler here for clarity of exposition.
\end{remark}

\paragraph{Winning Conditions}
Fix a candidate $c$. Under an anonymous and homogeneous voting rule, whether $c$ wins is a Boolean function $\Phi$ of a society $\vemu\in\R_{\ge0}^{m!}$, whose coordinates correspond to rankings in $\mathcal R$. We call $\Phi$ a \emph{Succinct Polyhedral Winner Regions (SPWR)} winning condition if there exist linear systems $A^\ell\vemu\le\veb^\ell$, $\ell=1,\dots,S(m)$, such that $\Phi(\vemu)=1$ if and only if $\vemu$ satisfies at least one of them. We require $S(m)$ to be polynomial in $m$. Each $A^\ell$ has $m^{\Oh(1)}$ rows and entries of encoding length $m^{\Oh(1)}$. We also require an algorithm that, given $\mathcal R'\subseteq\mathcal R$, computes the restrictions of all $A^\ell$ to the columns in $\mathcal R'$ in time $\poly(|\mathcal R'|,\tau,m)$. A voting rule is SPWR if the winning condition of every candidate is SPWR.

\paragraph{Examples and non-examples}
Every positional scoring rule is SPWR with a single region ($m-1$ inequalities $\mathrm{score}(c)\ge \mathrm{score}(d)$, each score linear in $\vemu$), and so is the weak-Condorcet rule ($m-1$ pairwise-majority inequalities with entries in $\{-1,0,1\}$).
For Bucklin we use a \emph{boundary-inclusive} variant and show it is SPWR.
Copeland, in contrast, is not SPWR.

\subsection{SPWR examples: Bucklin and Copeland}\label{app:bucklin}

\paragraph{Bucklin}
For a level $\ell \in [m]$, let $\mathrm{supp}_\ell(c)$ denote the total population ranking $c$ among its top $\ell$ candidates.
The \emph{majority threshold} is $n / 2$ in the discrete setting, and $\|\vemu\|_1 / 2$ in the continuous setting.

The \emph{Bucklin winning level} is the smallest level $\ell \in [m]$ such that $\mathrm{supp}_\ell(c)$ exceeds the majority threshold for some $c \in C$, i.e., $\mathrm{supp}_\ell(c) > \|\vemu\|_1 / 2$.
Notice that both sides of this expression are linear in $\vemu$.
A \emph{Bucklin winner} is then any candidate maximizing $\mathrm{supp}_\ell$ at the winning level $\ell$.
In the boundary-inclusive variant of Bucklin's rule, $c$ is a winner if and only if there is a level $\ell \in [m]$ such that (i) $\mathrm{supp}_\ell(c)$ reaches the majority threshold, (ii) $\mathrm{supp}_{\ell'}(d)$ does not exceed the threshold for any $d \in C$ and $\ell' < \ell$, and (iii) $\mathrm{supp}_\ell(c) \ge \mathrm{supp}_\ell(d)$ for every $d \in C$.

This variant is SPWR, with one system per level $\ell \in [m]$: condition (i) contributes a single inequality, and conditions (ii) and (iii) one inequality per candidate $d \in C$ each, so each system has $\Oh(m)$ rows.
For (ii) the instance at $\ell' = \ell-1$ suffices, since $\mathrm{supp}_{\ell'}(d)$ is nondecreasing in $\ell'$.
Each $\mathrm{supp}_\ell(c)$ is the sum of the coordinates $\mu_\rho$ over the rankings $\rho$ placing $c$ in one of the first $\ell$ positions, and the threshold is $\frac{1}{2}\sum_{\rho} \mu_\rho$, so every coefficient lies in $\{0, \pm\frac{1}{2}, \pm 1\}$.
Restricting a system to a given set $\mathcal{R}'$ of rankings needs only the position of each candidate in each $\rho \in \mathcal{R}'$.
The variant coincides with the standard rule whenever no support equals the threshold exactly.
The standard strict-majority rule itself fails to be SPWR, because its winner regions are not closed and a minimum-cost move need not exist.
Consider the society $(\frac12 [c^* \succ b \succ a], \frac12 [a \succ b \succ c^*])$:
under unit-cost Shift Bribery, shifting $c^*$ up in any $\varepsilon > 0$ of the second type's mass makes $c^*$ the unique strict-majority Bucklin winner at cost $2\varepsilon$, so no minimum-cost winning move exists.

\paragraph{Why Copeland is excluded}
Copeland is not SPWR.
Write $W_{c^*} = \{\vemu : c^* \text{ wins in } \vemu\}$, and call societies $\vemu_1, \dots, \vemu_L \in W_{c^*}$ a \emph{splitting family} for $c^*$ if $\vemu_i + \vemu_j \notin W_{c^*}$ whenever $i \neq j$.
Suppose $W_{c^*}$ is a union of convex sets, each contained in $W_{c^*}$.
Then no two members of a splitting family lie in a common such set, because then such a set would contain $\frac{1}{2}(\vemu_i + \vemu_j)$, and $W_{c^*}$ is a cone by homogeneity, so it would contain $\vemu_i + \vemu_j$.
A splitting family of size $L$ therefore forces at least $L$ convex sets in the SPWR description.
Now let $m = 2r+2$, and let the $2r+1$ candidates other than $c^*$ form a regular tournament in which every margin has absolute value $1$.
For a set $S$ of $r+1$ of them, let $c^*$ beat each candidate in $S$ by $1$ and lose to each candidate outside $S$ by $2$.
Every such vector of margins is the vector of margins of a nonnegative society~\citep{McGarvey1953}.
In that society $c^*$ scores $r+1$, as does every candidate outside $S$, so $c^*$ is a co-winner.
In the sum of the societies of two distinct sets $S$ and $S'$, candidate $c^*$ scores $|S \cap S'| \le r$, while every candidate outside $S \cap S'$ still scores $r+1$.
These $\binom{m-1}{m/2}$ societies thus form a splitting family.
For every even $m$ the region $W_{c^*}$ therefore needs $2^{m - \Oh(\log m)}$ convex pieces, which exceeds the $\poly(m)$ that SPWR allows.

\subsection{Relationship of SPWR to hyperplane rules and generalized scoring rules}
Hyperplane rules are defined by selecting finitely many affine hyperplanes in the simplex and requiring the rule to be constant on each cell of the induced arrangement.  This class coincides with generalized scoring rules (GSRs)~\cite{XiaConitzer2008GSR,MosselProcacciaRacz2013smooth}.

In a GSR, each vote $\prec\in\mathcal R$ is mapped to a $k$-dimensional score vector $f(\prec)\in\mathbb R^k$, these are summed over voters, and the winner is a function $g$ of the \emph{comparison pattern} among the coordinates of the total score vector $\sum_{i \in [n]} f(\prec_i)$.
Such rules are anonymous and homogeneous, and their winner regions in the simplex are unions of cells cut out by hyperplanes of the form ``coordinate $p$ equals coordinate $q$'' of the total score vector.

SPWR is narrower than hyperplane rules: even if a hyperplane rule is defined by only $\mathrm{poly}(m)$ hyperplanes, a candidate's winner region may be the union of an exponential number of arrangement cells.

\citet{XiaConitzer2009FLC} define \emph{finite local consistency (FLC)} and the \emph{degree of consistency} of a rule as the minimum number of parts in a partition of the profile space such that the rule is consistent on each part.  They show that anonymity together with FLC characterizes generalized scoring rules, and they compute/estimate degrees for many common rules (e.g., scoring rules have degree $1$, Bucklin has degree $\Theta(m)$, and several Condorcet-type rules have very large degrees).

The SPWR condition can be viewed as a \emph{geometric strengthening} of having polynomial degree of consistency: it demands not only that the rule be describable by polynomially many ``cases,'' but that each case correspond to an explicitly given polyhedron in the vote simplex with polynomial description complexity.

\subsection{The Continuous Minimum Move Problem}
\niceproblem{\textsc{Minimum Move}$_\infty$}
{A set of $\tau$ voter types $\{(\sigma_i, \pi_i)\}_{i \in [\tau]}$, a society $\vemu \in \mathbb{Q}^\tau_{\geq 0}$, a preferred candidate $c^* \in C$, a winning condition $\Phi: \mathbb{Q}^{m!}_{\geq 0} \to \{0,1\}$.}
{A move $\vex$ of minimum cost $\sum_{i \in [\tau]} \sum_{\rho \in \mathcal{R}_\emptyset} \pi_i(\rho) x_{i \to \rho}$ such that $\Phi(\venu) = 1$ in the society $\venu$ resulting from applying $\vex$ to $\vemu$.}

\subsection{Linear Programming with Oracles}
Both our algorithms and our hardness proofs concern linear programs of exponential size, which we handle with machinery of \citet{GLS} that is standard in combinatorial optimization but so far uncommon in ComSoC.
We use standard facts from linear algebra and linear programming without further attribution; see \cite[Section~0.1]{GLS} for this background and \cite[Section~6.2]{GLS} for the encoding-length bounds for rational polyhedra used below.

A \emph{linear program (LP) over a polyhedron} $P \subseteq \R^N$ asks to maximize a linear objective $\vecc^\top \vex$ over $P$; \emph{solving} it means returning an optimal solution and its value, or asserting infeasibility, or asserting unboundedness together with a direction along which the objective increases.
The \emph{separation problem} for $P$ is: given $\vey \in \mathbb{Q}^N$, either assert $\vey \in P$, or return a linear inequality valid for $P$ yet violated by $\vey$ (a ``separating hyperplane'' with $P$ on one side and $y$ on the other).
Finally, $P$ is \emph{well-described} if it comes with a bound $\phi$ on the encoding length of the inequalities needed to describe it, and running times are measured in $N$, $\phi$, and the encoding length of the input.
The pivotal fact is that optimizing and separating are two sides of the same coin:

\begin{proposition}[{\cite[Theorem~6.4.9]{GLS}}]\label{prop:gls-equivalence}
	For well-described polyhedra, optimization and separation are polynomially equivalent.
\end{proposition}

The direction ``separation $\Rightarrow$ optimization'' powers our positive results, and the construction always goes along the following template.
A \emph{configuration LP} is an LP with exponentially many variables -- here, one per pair (type, destination ranking) -- but few constraints.
Its \emph{dual} then has few variables and one constraint per configuration, and the separation problem of the dual is called the \emph{pricing problem}: find a configuration whose constraint a given dual solution violates, if one exists.
Pricing is a concrete combinatorial problem -- for us, always an optimization over rankings -- and once solved, the rest comes for free:

\begin{proposition}[{primal recovery,~\cite[Lemma~6.5.15]{GLS}}]\label{prop:gls-primal-recovery}
	Consider an LP with (possibly exponentially) many variables but only $N$ constraints besides nonnegativity, all of polynomial encoding length, whose pricing problem is solvable in polynomial time.
	If the LP has a finite optimum, then an optimal solution with at most $N$ nonzero variables can be computed in time polynomial in $N$ and the encoding lengths involved.
\end{proposition}

Where pricing is tractable, Proposition~\ref{prop:gls-primal-recovery} yields a polynomial-time algorithm.
Where it is hard, Section~\ref{sec:hardness} runs Proposition~\ref{prop:gls-equivalence} in reverse: a fast bribery algorithm would optimize over a certain dual polyhedron, hence also separate over it -- which we prove hard.

\section{Positive Results}\label{sec:positive}

\paragraph{Succinct Winner Regions: a Polynomial-Size LP}

In the continuum the moves themselves form a polytope, but the societies where $c^*$ wins are possibly \emph{not} convex, and minimizing over a union of polyhedra is not an LP.
SPWR is exactly the hypothesis that this union is short.

\begin{theorem} \label{thm:poly-lp}
	Let $\Phi$ be a SPWR winning condition.
	\textsc{Minimum Move}$_\infty$ with winning condition $\Phi$, with all populations and finite costs given as rationals, is solvable in time $\poly(m, \tau, T, L)$, where $T$ is the number of allowed output rankings and $L$ is the encoding length of the input.
	We assume that the finite-cost type--destination pairs can be enumerated in time $\poly(m,\tau,T,L)$.
\end{theorem}

\begin{remark}
Theorem~\ref{thm:poly-lp} uses the regions only through the LP.
It therefore still holds if a region is given by a polynomial-size \emph{extended formulation} rather than by an explicit linear system, with the auxiliary variables added to the LP.
This may extend its applicability to a larger class of voting rules.
\end{remark}

\begin{proof}[Proof of Theorem~\ref{thm:poly-lp}]
	Let
	\[
		\mathcal{T}:=
		\{\rho\in\mathcal{R}_{\emptyset}:
		\pi_i(\rho)<+\infty\text{ for some }i\in[\tau]\}.
	\]
	Thus $|\mathcal{T}|=T$.
	The set $\mathcal{T}$ may contain $\emptyset$, but the SPWR systems have columns only for active rankings in $\mathcal{R}$; this is why we restrict them below to $\mathcal{T}\cap\mathcal{R}$.
	For each $\ell\in[S(m)]$, we construct an LP for the minimum-cost move whose output society belongs to the region $A^\ell\venu\le\veb^\ell$.

	For every $i\in[\tau]$ and $\rho\in\mathcal{T}$ such that $\pi_i(\rho)<+\infty$, introduce a nonnegative variable $x_{i\to\rho}$.
	For each achievable ranking $\rho\in\mathcal{T}\cap\mathcal{R}$, define
	\[
		\nu_\rho:=
		\sum_{\substack{i\in[\tau]\\\pi_i(\rho)<+\infty}}x_{i\to\rho}.
	\]
	Mass sent to $\emptyset$ does not enter $\venu$, because inactive voters do not affect the winner.
	The LP for region $\ell$ is
	\[
		\begin{aligned}
		\min\quad&
			\sum_{i\in[\tau]}
			\sum_{\substack{\rho\in\mathcal{T}\\\pi_i(\rho)<+\infty}}
				\pi_i(\rho)x_{i\to\rho}\\
		\text{s.t.}\quad&
			\sum_{\substack{\rho\in\mathcal{T}\\\pi_i(\rho)<+\infty}}
				x_{i\to\rho}=\mu_i
				&& (i\in[\tau]),\\
		&A^\ell_{\mathcal{T}\cap\mathcal{R}}
			\venu_{\mathcal{T}\cap\mathcal{R}}\le\veb^\ell.
		\end{aligned}
	\]
	Here $A^\ell_{\mathcal{T}\cap\mathcal{R}}$ is the restriction supplied by the SPWR representation.

	Every finite-cost feasible move gives a feasible solution of this LP whenever its output society belongs to region $\ell$.
	Conversely, every feasible LP solution defines such a move by assigning zero mass to all omitted type--destination pairs.
	The objective is the cost of the move.
	Therefore, the LP optimum is exactly the minimum cost of reaching region $\ell$.

	There are at most $\tau T$ variables, $\tau$ mass equalities, and $\poly(m)$ winner-region inequalities.
	All coefficients have encoding length polynomial in $L$.
	The assumed enumeration algorithm constructs the LP in time $\poly(m,\tau,T,L)$, and the LP can be solved within the same polynomial bound.
	Finally, $S(m)$ is polynomial in $m$, so taking the minimum over all regions gives the claimed running time.
\end{proof}

With unit costs an optimal move only shifts $c^*$ up -- folklore for Condorcet, an exchange argument for Borda -- leaving each type at most $m$ allowed outputs, so $T \le m\tau$.
For Borda, both assumptions (unit costs across \emph{pairs} and \emph{types}) are necessary in the exchange argument; the slightly more general uniform cost regime is open (Open Problem~\ref{op:uniform}).

\begin{corollary} \label{cor:spwr-applications}
	\{Condorcet, Score, Bucklin\}-\{\textsc{Shift Bribery}, \textsc{CCDV}\}$_\infty$, Condorcet-\textsc{\$Bribery}$_\infty$, unit cost \{Borda, Condorcet\}-\textsc{Swap Bribery}$_\infty$ are in $\poly(m,\tau,L)$ time.
\end{corollary}

Call $\rho$ a \emph{shift} (of $c^*$ up from $\sigma_i$) if $\rho$ moves $c^*$ some number of positions towards the top of $\sigma_i$, leaving all other candidates' relative order unchanged; the shift by $t$ positions costs exactly $t$.

\begin{lemma}[Shift-only optimum] \label{lem:unit-cost-shift-only}
	Every instance of unit cost Borda-\textsc{Swap Bribery}$_\infty$ has an optimal move $\vex$ such that every ranking $\rho$ with $x_{i \to \rho} > 0$ is a shift of $c^*$ up from $\sigma_i$ (or $\sigma_i$ itself).
\end{lemma}

\begin{proof}[Proof of Lemma~\ref{lem:unit-cost-shift-only}]
	\citet[Theorem~15]{BaumeisterHR19} use precisely this exchange argument to prove that discrete unit-cost Borda-\textsc{Swap Bribery} and unit-price Borda-\textsc{Shift Bribery} have the same optimal budgets. Vote by vote, a successful set of swaps is replaced by shifting $c^*$ up by the number of swaps spent in that vote, capped at $c^*$'s distance from the top; this only improves $c^*$'s margins.\footnote{The proof appears as Theorem~12 of the COMSOC~'18 version of~\cite{BaumeisterHR19}, available at \url{https://people.cs.rutgers.edu/~lirong.xia/COMSOC18/papers/paper\%2055.pdf}; the AAAI'19 version states the result with the proof omitted.}
	The argument is oblivious to voter multiplicities, so it applies verbatim to a society continuum with divisible mass.
\end{proof}

\begin{proof}[Proof of Corollary~\ref{cor:spwr-applications}]
	Every scoring protocol, Condorcet's rule, and Bucklin's rule are SPWR, so by Theorem~\ref{thm:poly-lp} it suffices to bound the number of allowed output rankings by $T = \Oh(m\tau)$ in each case.
	For \textsc{Shift Bribery} and \textsc{CCDV} this is immediate from the cost functions: only the at most $m$ shifts of $c^*$ up from $\sigma_i$, respectively only $\sigma_i$ and $\emptyset$, have finite cost.
	For Condorcet-\textsc{\$Bribery}, at most two output rankings per type suffice. Unbribed mass stays at $\sigma_i$. Any bribed mass may be sent to the ranking obtained from $\sigma_i$ by moving $c^*$ to the first position, because this weakly maximizes all pairwise margins of $c^*$ at the same flat cost.
	For unit cost \textsc{Swap Bribery} an optimal move only shifts $c^*$ up, leaving $\textrm{rank}(c^*,\sigma_i) \le m$ outputs per type.
	For Condorcet this is the classical fact that the Dodgson score~\cite{bartholdiToveyTrick1989Dodgson} equals the unit-cost shift-bribery cost, see~\cite{elkindFaliszewski2010campaign}; for Borda it is Lemma~\ref{lem:unit-cost-shift-only}.
\end{proof}

\begin{remark}
	As discussed earlier, the problems Condorcet-\{\textsc{CCDV}, \textsc{Swap Bribery}, \textsc{Bribery}\} are exactly the homogeneous limits of the score problems for Young, Dodgson, and the Voter Replacement rules, respectively.
\end{remark}

\paragraph{Configuration LPs: Score-\textsc{\$Bribery}$_\infty$}

Theorem~\ref{thm:poly-lp} does not apply to \textsc{\$Bribery}, where a bribed voter may be moved to \emph{any} of the $m!$ rankings, so $T$ is \emph{a priori} unbounded.

We demonstrate the configuration-LP template on Score-\textsc{\$Bribery}$_\infty$,
where the pricing problem is the simplest.
The primal, dual, and pricing problem below will then repeat for other problems studied here, and we will only record how they change.

\begin{theorem} \label{thm:score-dollar-bribery}
	For every scoring protocol $\ves$, Score-\textsc{\$Bribery}$_\infty$ with all populations and prices given as rationals is solvable in time $\poly(m, \tau, L)$.
\end{theorem}

The pricing problem will reduce to sorting the opponents by their dual variables.

\begin{proof}[Proof of Theorem~\ref{thm:score-dollar-bribery}]
	The primal is the LP of Theorem~\ref{thm:poly-lp} with all $m!$ destinations allowed: minimize $\sum_{i \in [\tau]} \sum_{\rho \in \mathcal{R}} \pi_i(\rho)\, x_{i\to\rho}$ subject to $\sum_{\rho \in \mathcal{R}} x_{i \to \rho} = \mu_i$ for every $i \in [\tau]$, the winning constraints $\sum_{i \in [\tau]} \sum_{\rho \in \mathcal{R}} x_{i\to\rho} \big(s_{\textrm{rank}(c^*,\rho)} - s_{\textrm{rank}(c,\rho)}\big) \ge 0$ for every $c \neq c^*$, and $\vex \ge \vezero$.
	Recall that in \textsc{\$Bribery}, $\pi_i(\sigma_i) = 0$ and $\pi_i(\rho) = \gamma_i$ for every $\rho \neq \sigma_i$.
	It is always feasible: fix any $\rho$ with $c^*$ first and set $x_{i \to \rho} = \mu_i$ for every type $i$. As $\ves$ is non-increasing, $c^*$ then has score $s_1$ per unit of mass and every opponent at most that, so $c^*$ wins.
	The dual has a free variable $\alpha_i$ per type, a variable $\zeta_c \ge 0$ per winning constraint, the objective $\max \sum_{i \in [\tau]} \mu_i \alpha_i$, and one constraint $\alpha_i + g(\rho) \le \pi_i(\rho)$ per column, where
	\[
		g(\rho) := \sum_{c \neq c^*} \zeta_c \big(s_{\textrm{rank}(c^*,\rho)} - s_{\textrm{rank}(c,\rho)}\big) .
	\]
	The pricing problem is to maximize $g$ over $\mathcal{R}$.
	Here the $\zeta_c$ are fixed and only the positions of the candidates in $\rho$ are ours to choose.
	Moving $c^*$ up one position raises $s_{\textrm{rank}(c^*,\rho)}$ and lowers the score of the single candidate it passes, so it weakly increases every difference $s_{\textrm{rank}(c^*,\rho)} - s_{\textrm{rank}(c,\rho)}$, and some optimal $\rho$ has $c^*$ first.
	The opponents then take positions $2, \dots, m$, whose differences $s_1 - s_2 \le \dots \le s_1 - s_m$ form a fixed non-decreasing vector, and maximizing $g$ is a matching of the $\zeta_c$ against it.
	A dot product of two vectors is largest when both are sorted the same way, so the largest $\zeta_c$ takes the largest difference, that is, the last position.

	Formally, setting $w_{c^*} := \sum_{c \neq c^*} \zeta_c$ and $w_c := -\zeta_c$ for $c \neq c^*$, we have $g(\rho) = \sum_{c \in C} w_c\, s_{\textrm{rank}(c,\rho)}$: an assignment of the score values to the weights.
	Since $\ves$ is non-increasing, the rearrangement inequality~\cite{HardyLittlewoodPolya1952} says that $g$ is maximized by any $\rho$ ranking the candidates in non-increasing order of $w_c$ -- that is, $c^*$ first and the opponents in non-decreasing order of $\zeta_c$ -- found by sorting.
	This yields the following separation oracle.
	First return any violated inequality $\zeta_c \ge 0$.
	Otherwise $\vezeta \geq \vezero$ and it remains to separate the constraints $\alpha_i + g(\rho) \le \pi_i(\rho)$.
	These come in two kinds for each type $i$, because $\pi_i$ is not constant: $\pi_i(\sigma_i) = 0$ whereas $\pi_i(\rho) = \gamma_i$ for every $\rho \neq \sigma_i$.
	So test $\alpha_i + g(\sigma_i) \le 0$ directly, and compare $\alpha_i + g(\rho^*)$ against $\gamma_i$, where $\rho^*$ is the sorted maximizer of $g$.
	If the latter exceeds $\gamma_i$ and $\rho^* \neq \sigma_i$, return the constraint of $\rho^*$.
	If it exceeds $\gamma_i$ and $\rho^* = \sigma_i$, then $\alpha_i + g(\sigma_i) > \gamma_i \ge 0$, so the direct test has already returned the constraint of $\sigma_i$.

	The dual polyhedron is well-described (every returned inequality has coefficients among $0$, $1$, score differences, and $\gamma_i$), nonempty ($(\vezero, \vezero)$ is feasible), and its objective is bounded, because the primal is a nonempty compact polytope with a nonnegative objective. Thus, by strong duality both optima are attained and equal.
	Hence the ellipsoid method with the above oracle finds an optimal dual solution, and Proposition~\ref{prop:gls-primal-recovery} then yields an optimal \emph{primal} solution in time $\poly(m, \tau, L)$.
	Besides nonnegativity the primal has $\tau + m - 1$ constraints, so this solution has at most $\tau + m - 1$ nonzero type--destination variables.
\end{proof}

Moving to \textsc{Swap Bribery} changes one thing: the price $\pi_i(\rho)$ now depends on the destination. Pricing thus no longer collapses to maximizing $g$, and its complexity now depends on the structure of the swap costs and the voting rule.

\paragraph{$k$-Approval-\textsc{Swap Bribery}$_\infty$}

Throughout, the parameter $k$ satisfies $1 \le k < m$ and is part of the input.
We begin with \emph{general} swap costs, giving an algorithm polynomial for every fixed $k$ (\textup{XP} in $k$). We show that, for any voter type, it suffices to consider only a bounded set of allowed output rankings, which we now define.
Consider a ranking $\rho$, a set $S\subseteq C$ with $|S|=k$, and a ranking $\rho'$ whose set of top-$k$ candidates is exactly $S$.
We say that $\rho'$ \emph{respects} $\rho$ if the relative order of any two candidates that are both in $S$, or both in $C \setminus S$, is the same in $\rho'$ as in $\rho$.
The unique ranking whose set of top-$k$ candidates is exactly $S$ and that respects $\rho$ is denoted by $\hat\rho_{S,\rho}$.
We write $\hat\rho_S$ when $\rho$ is clear from the context, and $\hat\rho_{S,i}$ for $\hat\rho_{S,\sigma_i}$. Finally set $\mathcal{T}^\sigma := \{\hat\rho_{S,\sigma} \mid S \subseteq C,\ |S|=k\}$.
For $S\subseteq C$ with $|S|=k$ we also write $\mathcal{T}_{S}$ for the set of rankings placing all of $S$ above all of $C\setminus S$; thus $\hat\rho_{S,\rho}$ is the unique member of $\mathcal{T}_S$ that respects $\rho$, and $\mathcal{T}^\sigma$ picks one member of each $\mathcal{T}_S$.

\begin{claim}\label{claim:minimim:cost:rankings:S}
    Let $C$ be a set of candidates, $\rho$ a ranking over $C$, and $\pi$ a swap-cost function. For any $S\subseteq C$ with $|S|=k$ and ranking $\rho' \in \mathcal{T}_S$, we have $\pi(\hat\rho_{S,\rho})\le \pi(\rho')$.
\end{claim}
\begin{proofclaim}
        Since all swap costs are nonnegative, $\pi(\rho')$ is the sum of the costs of the pairs of
        candidates whose relative order differs between $\rho$ and $\rho'$.
        For every candidate $c\in S$ and every candidate $c'\in C\setminus S$, the relative order
        between $c$ and $c'$ is determined by the fact that all candidates in $S$ must appear above
        all candidates in $C\setminus S$.
        Hence, the contribution of such pairs depends only on the set $S$, and not on the order
        inside $S$ or inside $C\setminus S$.
        On the other hand, any inversion inside $S$ or inside $C\setminus S$ contributes an
        additional nonnegative swap cost. Thus, the minimum-cost ranking with top-$k$ set $S$ is
        obtained by preserving the relative order inside both blocks, namely $\hat{\rho}_S$.
\end{proofclaim}

\begin{corollary} \label{thm:XP:algo:k-appr}
    $k$-Approval Margin-\textsc{Swap Bribery}$_\infty$ can be solved in $m^{\Oh(k)} \poly(L)$ time.
\end{corollary}

\begin{proof}[Proof of Corollary~\ref{thm:XP:algo:k-appr}]
	Let $C$ be a candidate set, and $\vemu \in \R_{\ge 0}^\tau$ a society over types $\{(\sigma_i,\pi_i)\}_{i \in [\tau]}$.
    Since the voting rule is SPWR, using Theorem~\ref{thm:poly-lp}, we can directly conclude that
    $k$-Approval Margin-\textsc{Swap Bribery}$_\infty$ can be solved in $\poly(m, \tau, T, L)$ time, where $T$ is the number of allowed output rankings. However, there is no guarantee that the number of allowed output rankings $T$ is bounded.

    We will now prove that we can bound the number of allowed output rankings by restricting, for each type $i \in [\tau]$, the allowed output rankings to the set
    $\mathcal{T}^{\sigma_i} = \{ \hat\rho_{S,i}\mid S \subseteq C\text{ and } |S|=k \}$, while preserving at least one optimum.

    \begin{claim} \label{k:approval:XP:claim:output::types}
        For any feasible move $\vex = (x_{i \to \rho})_{i \in [\tau], \rho \in \mathcal{R}}$, there exists a feasible move $\vex'$ such that:
        \begin{itemize}
            \item the output societies of $\vex$ and $\vex'$ assign every candidate the same $k$-Approval score (in particular, $\vex$ makes $c^*$ a winner with margins $\veDelta$ if and only if $\vex'$ does),
            \item for every $i \in [\tau]$, $x'_{i\rightarrow \rho} = 0$ for all $\rho \notin \mathcal{T}^{\sigma_i}$, and
            \item $\sum_{i\in [\tau]} \sum_{\rho \in \mathcal{R}} x'_{i\rightarrow \rho } \pi_i(\rho) \le \sum_{i\in [\tau]} \sum_{\rho \in \mathcal{R}} x_{i\rightarrow \rho } \pi_i(\rho)$.
        \end{itemize}
    \end{claim}

    \begin{proofclaim}
        Define $\vex'$ simultaneously for all $i \in [\tau]$ and all $S \subseteq C$ with $|S| = k$ by aggregating each set $\mathcal{T}_S$ into its stable representative:
        \[
            x'_{i\rightarrow \hat \rho_{S,i}} := \sum_{\rho\in \mathcal{T}_S} x_{i\rightarrow \rho},
            \qquad
            x'_{i\rightarrow \rho} := 0 \quad \text{for } \rho \in \mathcal{T}_S\setminus \{\hat \rho_{S,i}\} .
        \]
        All rankings in $\mathcal{T}_S$ assign identical $k$-Approval points to every candidate, so every candidate's score is unchanged, and the row sums are preserved, so $\vex'$ is feasible.
        By Claim~\ref{claim:minimim:cost:rankings:S}, $\pi_i(\hat \rho_{S,i})\le \pi_i(\rho)$ for every $\rho \in \mathcal{T}_S$, so the cost does not increase.
    \end{proofclaim}

    Due to Claim~\ref{k:approval:XP:claim:output::types}, we can restrict the LP of Theorem~\ref{thm:poly-lp} to the allowed output rankings $\bigcup_{i \in [\tau]} \mathcal{T}^{\sigma_i}$ (with the winning constraints instantiated as the margin constraints, i.e., $\sum_{i \in [\tau]} \sum_{\rho\in \mathcal{T}^{\sigma_i}} (I_{c^*,\rho}- I_{c,\rho})x_{i\rightarrow \rho} \ge \Delta_c$ for all $c \neq c^*$, where $I_{c,\rho}=1$ if $\CandRank{\rho}{c}\le k$ and $I_{c,\rho}=0$ otherwise).
    Since, for each $i\in [\tau]$, we have $|\mathcal{T}^{\sigma_i}|=\binom{m}{k} = \Oh(m^k)$, the LP has $\tau \binom{m}{k}$ variables and $\Oh(\tau+m)$ constraints. Consequently, it can be constructed and solved in time $\poly\bigl(\tau\binom{m}{k},m,L\bigr)=m^{\Oh(k)}\poly(L)$.
\end{proof}

\paragraph{Additively separable costs}
To get a polynomial time algorithm for a narrower class of costs, we consider again the configuration LP template established previously.
The primal--dual pair is the one of Theorem~\ref{thm:score-dollar-bribery} with two changes: the columns of type $i$ range over its stable rankings $\mathcal{T}^{\sigma_i}$ rather than all of $\mathcal{R}$, and the margin coefficients are $d_{\rho,c} = I_{c^*,\rho} - I_{c,\rho}$, where $I_{c,\rho} := 1$ if $\CandRank{\rho}{c}\le k$ and $I_{c,\rho} := 0$ otherwise.
Writing $\zeta_c \ge 0$ for the dual variable of opponent $c$, the pricing problem of type $i$ therefore reads as follows.

\niceproblem{$k$-\textsc{Approval-Swap Pricing}}
{A set $C$ of size $m$, a distinguished candidate $c^*\in C$, a ranking $\sigma \in \mathcal{R}$, rational pairwise swap costs $\delta(c,c') \ge 0$, and a vector $(\zeta_c)_{c \neq c^*}\in \mathbb{Q}_{\ge 0}^{m-1}$.}
{A ranking $\rho^* \in \mathcal{T}^\sigma := \{\hat\rho_{S,\sigma} \mid S \subseteq C,\ |S|=k\}$ maximizing the value $\sum_{c \neq c^*} (I_{c^*,\rho^*}- I_{c,\rho^*}) \zeta_c - \pi(\rho^*)$.}

We show that pricing is tractable for \emph{additively separable} swap costs: whenever $c$ and $c'$ occupy positions $j < j'$ of $\sigma_i$, their swap cost is $\delta_i(c,c') = a^i_j + b^i_{j'}$. The vectors $\vea^i, \veb^i \in \mathbb{Q}^m$ with $a^i_j+b^i_{j'}\ge 0$ are part of the input.

\begin{theorem} \label{thm:k-approval-unit-separation}
    With additively separable swap costs, $k$-\textsc{Approval-Swap Pricing} can be solved in polynomial time.
\end{theorem}

Before giving the full proof, we sketch the intuition behind why separable additive costs make things easier. Relabel the candidates $c_1, \dots, c_m$ in the order of $\sigma$, most preferred first.
Within a fixed top-$k$ set $S$ every ranking earns the same profit, and $\hat\rho_S$ is the cheapest of them, so the search is over sets rather than over rankings.
The difficulty is then the cost: reaching $\hat\rho_S$ inverts exactly the pairs $j<j'$ with $c_j \notin S \ni c_{j'}$, which couples every chosen candidate to every unchosen candidate preceding it in $\sigma$.
Additive separability dissolves the coupling, since a pair's cost $a_j+b_{j'}$ splits into a term $a_j$ charged to $c_j$ and a term $b_{j'}$ charged to $c_{j'}$.
A chosen candidate $c_i$ then pays only for the unchosen candidates preceding it, and an unchosen one only for the chosen candidates following it; both counts are determined by how many of $c_1, \dots, c_i$ lie in $S$, which is what the dynamic program below remembers.

\begin{proof}[Proof of Theorem~\ref{thm:k-approval-unit-separation}]
    Fix an instance and relabel the candidates so that $\sigma=c_1\prec c_2\prec \cdots \prec c_m$, with the given vectors $(a_1,\ldots,a_m),(b_1,\ldots,b_m)$.
    For each set $S\subseteq C$ of size $k$, the value $\sum_{c \neq c^*} (I_{c^*,\rho}- I_{c,\rho}) \zeta_c$ is the same for every ranking $\rho$ with top-$k$ set $S$, and among those, $\hat{\rho}_S$ minimizes the cost (Claim~\ref{claim:minimim:cost:rankings:S}); so, like before, it suffices to optimize over the rankings $\hat\rho_S$.

    Let $S=\{c_{i_1},\ldots,c_{i_k}\}$ where $i_1<i_2<\cdots<i_k$.
    In the ranking $\hat{\rho}_S$, the only inversions with respect to $\sigma$ are pairs
    $(c_j,c_{j'})$ such that $j<j'$, $c_j\notin S$, and $c_{j'}\in S$.
    Therefore
    \begin{equation} \label{eq:k-approval-additive-cost}
        \pi(\hat{\rho}_S)=
        \sum_{\substack{j<j'\\ c_j\notin S,\ c_{j'}\in S}}(a_j+b_{j'}).
    \end{equation}

    Next, set $w_{c^*} := \sum_{c \neq c^*} \zeta_c$ and $w_c := -\zeta_c$ for $c \neq c^*$.
    Then, for every set $S$ of size $k$,
    \[
        \sum_{c \neq c^*} (I_{c^*,\hat{\rho}_S}- I_{c,\hat{\rho}_S}) \zeta_c
        =
        \sum_{c\in S} w_c ,
    \]
    so, combining with~\eqref{eq:k-approval-additive-cost}, the objective to maximize over $S$ is
    \[
        \sum_{c\in S} w_c
        -\sum_{\substack{j<j'\\ c_j\notin S,\ c_{j'}\in S}}(a_j+b_{j'}).
    \]

    This objective admits an $\Oh(mk)$ dynamic program.
    For $i\in \{0,\ldots,m\}$ and $t\in \{0,\ldots,k\}$, let $\mathrm{DP}[i,t]$ denote the
    maximum value obtainable from the first $i$ candidates $\{c_1,\ldots,c_i\}$ if exactly $t$ of
    them are chosen into the top-$k$ set.
    We set $\mathrm{DP}[0,0]=0$ and $\mathrm{DP}[0,t]=-\infty$ for each $t>0$.

    Consider a state $(i,t)$ with $i\ge 1$.
    If we select candidate $c_i$, then $c_i$ becomes the $t$-th selected candidate.
    Among the previous $i-1$ candidates, exactly $t-1$ were selected, so exactly
    $(i-1)-(t-1)=i-t$ were not selected.
    Each of these unselected predecessors forms an inversion with $c_i$ and contributes $b_i$.
    Hence the transition value is
    $\mathrm{DP}[i-1,t-1] + w_{c_i} - b_i(i-t)$.
    If we do not select $c_i$, then exactly $k-t$ selected candidates must still appear to its
    right in the final top-$k$ set.
    Each such later selected candidate forms an inversion with $c_i$ and contributes $a_i$.
    Hence the transition value is
    $\mathrm{DP}[i-1,t] - a_i(k-t)$.
    Therefore, for all feasible $i,t$,
    \begin{align*}
        \mathrm{DP}[i,t]=
        & \max \{
            \mathrm{DP}[i-1,t]- a_i(k-t), \\
            & \mathrm{DP}[i-1,t-1]+w_{c_i}- b_i(i-t)
        \},
    \end{align*}
    where the second term is considered only when $t\ge 1$.

    Correctness rests on the following exact rewriting: for a set $S$ of size $k$, writing $s_i := \mathbf{1}[c_i \in S]$ and $t_i := \sum_{j \le i} s_j$, the objective equals
    \[
        \sum_{i=1}^m \Big( s_i\, w_{c_i} - s_i\, b_i (i - t_i) - (1 - s_i)\, a_i (k - t_i) \Big),
    \]
    since $b_i$ is charged once per unselected predecessor of a selected $c_i$ (there are $i - t_i$ of them) and $a_i$ once per selected successor of an unselected $c_i$ (there are $k - t_i$).
    Defining $\mathrm{DP}[i,t]$ as the maximum of the first $i$ summands over choices with $t_i = t$ yields the displayed recurrence, and the value of an optimal top-$k$ set is $\mathrm{DP}[m,k]$.
    After filling the table, we reconstruct an optimal set $S$ by backtracking through the chosen
    transitions and output the corresponding stable ranking $\hat{\rho}_S$.

    The DP table has $\Oh(mk)$ states and each state is processed in constant time.
    Hence the pricing problem is polynomial-time solvable.
\end{proof}

\begin{corollary} \label{cor:k-approval-additive-poly}
    $k$-Approval Margin-\textsc{Swap Bribery}$_\infty$ with additively separable swap costs is polynomial time solvable.
\end{corollary}

\begin{proof}[Proof of Corollary~\ref{cor:k-approval-additive-poly}]
    Consider the configuration LP of Corollary~\ref{thm:XP:algo:k-appr}, restricted for each $i\in[\tau]$ to variables $x_{i \to \rho}$ with $\rho \in \mathcal{T}^{\sigma_i}$.
    Its dual has $\tau + m - 1$ variables, and its separation problem decomposes into checking the nonnegativity constraints directly and one instance of $k$-\textsc{Approval-Swap Pricing} per voter type (the free per-type dual variable $\alpha_i$ is an additive constant in type $i$'s constraints).
    Each of these calls is solvable in polynomial time by Theorem~\ref{thm:k-approval-unit-separation}.
    All coefficients returned by this oracle have polynomial encoding length, so the dual polyhedron is well-described.
    The ellipsoid method solves the dual (while also detecting unboundedness, which certifies primal infeasibility) and the primal recovery (Proposition~\ref{prop:gls-primal-recovery}) then yields an optimal primal solution with at most $\tau+m-1$ nonzero variables, i.e., an explicit optimal move.
\end{proof}

\section{Hardness Results}\label{sec:hardness}

We first show that the tractability of additively separable costs in the society continuum does not extend to discrete electorates.
The remainder of the section concerns general swap costs in the society continuum.
Our positive results follow the route ``pricing is easy $\Rightarrow$ bribery is easy.''
For the continuous hardness results we establish the converse, ``pricing is hard $\Rightarrow$ bribery is hard'' -- which is not automatic, since the trivial implication goes the other way.
The bridge (Lemma~\ref{lem:bribery-to-separation}) turns a hypothetical bribery algorithm for merely the decision problem into a separation oracle for the dual polyhedron $K$.
It then suffices to prove the hardness of separation over $K$.
We do so for Borda (NP-hardness) and $k$-Approval ($W[1]$-hardness, parameter-preservingly).

\subsection{A Discrete Contrast for Additively Separable Costs}

Recall that the input for an additively separable instance supplies two arrays $\vea^v,\veb^v$ for every vote $v$.
If candidates occupy positions $i<j$ in the original vote, then swapping this pair costs $a_i^v+b_j^v$.
The positions in this formula always refer to the original vote, as in the standard static pair-cost model for \textsc{Swap Bribery}.

\begin{theorem}\label{thm:discrete-kapproval-separable}
For every fixed $k\ge 2$, constructive $k$-Approval-\textsc{Swap Bribery} with additively separable nonnegative rational costs is NP-complete. This holds even when every $a_i^v\in\{0,1\}$ and every $b_i^v\in\{1,2\}$.
\end{theorem}

\begin{proof}Membership in NP follows from the standard static pair-cost formulation.
A certificate specifies the final profile.
For each vote, the minimum cost of reaching its final order is the sum of the costs of the pairs inverted between the original and final orders.
Every such pair must be swapped, while an adjacent-swap sorting sequence swaps each inverted pair exactly once.
Hence, the total cost and the resulting $k$-Approval scores can be checked in polynomial time.

We first fix a $k\ge 2$. We then modify the reduction of \citet[Theorem~3]{DornSchlotter2010COMSOC} from \textsc{Multicolored Clique}.

We start by creating exactly the same $m$ candidates and $n$ voters as in \citet[Theorem~3]{DornSchlotter2010COMSOC}.

In addition, we create $(k-2)n$ dummy candidates $y_1,\ldots,y_{(k-2)n}$. For each voter, we extend the ranking by placing $k-2$ of the new dummy candidates at the top and the remaining dummy candidates at the bottom. Moreover, we have sufficiently many dummy candidates to ensure that the $k-2$ candidates placed at the top are distinct for every voter. In particular, none of these candidates can obtain a higher score than our preferred candidate.

We denote each vote $v$ by
\[
    x_1^v\succ x_2^v\succ \ldots \succ x_{m'}^v,
\]
where $m'=m+(k-2)n$.

Recall that, in the original proof, all swaps have cost $1$, except for two pairs in each special selection or consistency vote. We use the same distinction here, with their special voters also being our special voters.

For every ordinary voter $v$, where all swaps have cost $1$, we set $a_i^v=0$ and $b_i^v=1$ for every $x_i^v$. Thus, all swaps still have cost $1$.

For every special voter $v$, we instead set
\[
 a_i^v=
 \begin{cases}
 1,&i=k+1,\\
 0,&i\ne k+1,
 \end{cases}
 \qquad
 b_j^v=
 \begin{cases}
 2,&j=k,\\
 1,&j\ne k.
 \end{cases}
\]

Before proceeding, we compare the swap costs between candidates that appear in both constructions. The original candidates now occupy positions $k-1$ to $m+(k-2)$. Let $i<j$, and let $c_i,c_j$ be two candidates that occupy the $i$-th and $j$-th positions, respectively, in the original voter. Their new swap cost is
\[
 a_{i+k-2}^v+b_{j+k-2}^v=
 \begin{cases}
 2,&(i,j)=(1,2),\\
 2,&i=3\text{ and }j\ge4,\\
 1,&\text{otherwise}.
 \end{cases}
\]
Thus, the two original cost-$2$ pairs retain their costs, and no pair becomes cheaper. The only newly expensive pairs are $(3,j)$ with $j\ge5$, which, in our new enumeration, correspond to $(k+1,j)$ with $j\ge k+3$.

It remains to argue that the swaps that need to be performed are exactly the same in the original construction and in ours. As in the original proof, since we have kept all swap costs at least $1$, the structure of the voters ensures that we never need to consider swapping a pair involving $x_i^v$ for any $i\ge k+3$ (corresponding to $i\ge5$ in the original proof).

Therefore, the only new candidates that need to be considered are the top $k-2$ dummy candidates.

Observe that all top $k-2$ candidates of every voter have a score lower than that of the preferred candidate. Hence, there is no reason to move any of them down in the ranking. Moreover, all of the budget in the original proof is spent on reducing the costs associated with the relevant candidates. Consequently, there is no reason to include any candidate appearing in the first $k-2$ positions in a swap. Thus, all swaps take place among the candidates in positions $k-1$ to $k+2$, corresponding exactly to the candidates occupying the top four positions in the original construction.

Finally, the swaps between these candidates have exactly the same costs in both constructions. Therefore, the reduction holds for our cost function.

\end{proof}

\subsection{The Bridge: Bribery to Separation}\label{sec:oracle-to-dual}

The bridge in this subsection turns an algorithm for the bribery decision problem into a separation oracle for the dual polyhedron used in both hardness proofs.
We first state the common primal and dual LPs.

Both hardness constructions use a single bribable voter type.
Let $\mathcal{Q}$ be the set of rankings indexing its possible destinations.
For Borda, $\mathcal{Q}=\mathcal{R}$; for $k$-Approval, $\mathcal{Q}=\mathcal{T}^{\sigma_1}$.
For each $\rho\in\mathcal{Q}$ and $c\ne c^*$, let $d_{\rho,c}$ be the score margin of $c^*$ over $c$ contributed by one unit of mass at $\rho$.
Thus
\[
	d_{\rho,c}=\textrm{rank}(c,\rho)-\textrm{rank}(c^*,\rho)
\]
for Borda, while $d_{\rho,c}=I_{c^*,\rho}-I_{c,\rho}$ for $k$-Approval.
Given a mass $\mu$ and required margins $\veDelta$, write $\veb:=(\mu,\veDelta)$.
The primal LP is
\[
	\begin{aligned}
	P(\veb):\qquad
	\min\quad &\sum_{\rho\in\mathcal{Q}}\pi(\rho)x_\rho\\
	\text{s.t.}\quad
	&\sum_{\rho\in\mathcal{Q}}x_\rho=\mu,\\
	&\sum_{\rho\in\mathcal{Q}}d_{\rho,c}x_\rho\ge\Delta_c
		&& (c\ne c^*),\\
	&\vex\ge\vezero,
	\end{aligned}
\]
Its dual feasible region is
\[
	\begin{aligned}
	K=\Big\{(\alpha,\vezeta):\quad
	&\alpha+\textstyle\sum_{c\ne c^*}d_{\rho,c}\zeta_c\le\pi(\rho)
		&&(\rho\in\mathcal{Q}),\\
	&\vezeta\ge\vezero\Big\}.
	\end{aligned}
\]
The dual objective for right-hand side $\veb$ is $\veb^\top\vey$, where $\vey=(\alpha,\vezeta)$.
The coefficients $d_{\rho,c}$ are integers of absolute value at most $m-1$.
Together with the succinct encoding of $\pi$ from Section~\ref{sec:prelim}, this makes $K$ well-described.

Write $D=(\mathcal{Q},d,\pi)$ for the fixed LP data and $L$ for the total input length.
Suppose that an algorithm $\mathcal{A}$ receives $D$, a feasible rational right-hand side $\veb$, and a rational budget $\beta$.
It decides whether the optimal value of $P(\veb)$ is at most $\beta$.
Standard bounds on rational LP optima allow the exact value $V(\veb)$ to be recovered with polynomially many calls to $\mathcal{A}$~\cite[proof of Theorem~6.4.9]{GLS}.
Below, a value query means this polynomial-time reduction to $\mathcal{A}$.

\begin{lemma}\label{lem:bribery-to-separation}
	The separation problem over $K$ is solvable with $\poly(L)$ calls to $\mathcal{A}$ and $\poly(L)$ additional time.
	Every call uses the same data $D$ and has encoding length $\poly(L)$.
	Consequently, a running time of $\poly(L)$, respectively $f(k)\poly(L)$, for $\mathcal{A}$ gives the same type of running time for separation over $K$.
\end{lemma}

\paragraph{Standard-form notation}
We use generic LP notation for the remainder of the argument.
Fix an ordering $c_1,\dots,c_{m-1}$ of $C\setminus\{c^*\}$ and introduce a slack variable $s_c\ge0$ for every margin constraint.
Let $\vev=(\vex,\ves)$ contain all primal variables.
Let $\vep$ be their objective-coefficient vector: its $x_\rho$-entry is $\pi(\rho)$ and its $s_c$-entry is $0$.
Then
\[
	P(\veb)=\min\{\vep^\top\vev:A\vev=\veb,\ \vev\ge\vezero\}.
\]
The column of $x_\rho$ in $A$ is
\[
	\vea_\rho=(1,d_{\rho,c_1},\dots,d_{\rho,c_{m-1}})^\top,
\]
and the column of $s_{c_j}$ is $\veh_{c_j}=(0,-\vece_j)^\top$.
Thus $A$ is integral, has rank $m$, and every entry has absolute value at most
\[
	M:=\max\{1,m-1\}.
\]

A basis $B$ consists of $m$ linearly independent columns of $A$.
Write $A_B$ for the resulting square matrix and $\vep_B$ for the corresponding objective coefficients.
For a right-hand side $\ver$, the basic variables are $A_B^{-1}\ver$.
The basis is feasible for $\ver$ if this vector is nonnegative, and it is nondegenerate if none of its coordinates is zero.

We first handle objectives that are unbounded over $K$.
This part uses the pricing structure of the bribery LP.
A recession direction of $K$ is a vector $\ved$ such that $\vey+\lambda\ved\in K$ for every $\vey\in K$ and every $\lambda\ge0$.
If $\veb^\top\ved>0$, then $\ved$ certifies that $\veb^\top\vey$ is unbounded above over $K$.

\begin{proposition}\label{prop:achievable-margins}
	The primal $P(\veb)$ is feasible if and only if $\veb^\top\vey$ is bounded above over $K$.
	Given any rational $\veb$, one can either decide that $P(\veb)$ is feasible or return a recession direction $\ved$ of $K$ such that $\veb^\top\ved>0$.
	This takes time polynomial in $m$ and in the encoding lengths of $\veb$ and the scoring vector, independently of the move costs $\pi$.
\end{proposition}

\begin{proof}
	Primal feasibility does not depend on the objective coefficients.
	Set every move cost to $0$ and denote the resulting dual feasible region by
	\[
		\begin{aligned}
		K_0=\Big\{(\alpha,\vezeta):\quad
		&\alpha+\textstyle\sum_{c\ne c^*}d_{\rho,c}\zeta_c\le0
			&&(\rho\in\mathcal{Q}),\\
		&\vezeta\ge\vezero\Big\}.
		\end{aligned}
	\]
	Comparing the definitions of $K$ and $K_0$, $\ved\in K_0$ exactly when $\vey+\lambda\ved\in K$ for every $\vey\in K$ and every $\lambda\ge0$.
	Thus $K_0$ is the recession cone of $K$~\cite[Section~0.2]{GLS}.
	By LP duality, the zero-cost primal is feasible exactly when $\veb^\top\vey$ is bounded above over $K_0$~\cite[Theorem~0.1.49]{GLS}.
	This is also exactly when the same objective is bounded above over $K$.

	Separation over $K_0$ is the zero-cost pricing problem solved in Theorem~\ref{thm:score-dollar-bribery}.
	The coefficient bounds above make $K_0$ well-described.
	Proposition~\ref{prop:gls-equivalence} therefore optimizes over $K_0$ in the claimed time.
	If the objective is unbounded, it returns a vector $\ved\in K_0$ with $\veb^\top\ved>0$.
	Since $K_0$ is the recession cone of $K$, this vector is the required certificate.
\end{proof}

The remaining argument converts values into an optimal dual vector.
The perturbation and sensitivity statements are generic.
The only use of the scoring rule is to verify that the chosen perturbed right-hand side remains feasible.

\begin{lemma}\label{lem:value-to-dual}
	Assume that $V(\ver)$ can be evaluated for every feasible rational right-hand side $\ver$.
	Given any rational $\veb$, one can either return a recession direction $\ved$ of $K$ with $\veb^\top\ved>0$, or return a point of $K$ maximizing $\veb^\top\vey$.
	The algorithm uses at most $m+1$ value queries on feasible right-hand sides of encoding length $\poly(L)$ and $\poly(L)$ additional time.
\end{lemma}

\begin{proof}
	Apply Proposition~\ref{prop:achievable-margins}.
	If $P(\veb)$ is infeasible, return the recession direction supplied by the proposition.
	We henceforth assume that $P(\veb)$ is feasible.
	Because $\vep\ge\vezero$, its optimal value is finite.

	We first clear denominators.
	Choose a positive integer $N$ of encoding length $\poly(L)$ such that $N\veb$ is integral.
	Replacing $\veb$ by $N\veb$ scales every primal and dual objective value by $N$ but does not change the set of dual maximizers.
	We rename the scaled right-hand side $\veb$ and assume from now on that $\veb\in\mathbb{Z}^m$.

	Let
	\[
		\Sigma:=\operatorname{diag}(1,-1,\dots,-1),
		\qquad
		\veeps:=(\varepsilon,\varepsilon^2,\dots,\varepsilon^m)^\top.
	\]
	Set
	\[
		q:=\left\lceil\log_2\big((m!)^2M^{2m-1}\big)\right\rceil+1,
		\qquad \varepsilon:=2^{-q}.
	\]
	Define $\veb':=\veb+\Sigma\veeps$.
	Thus $\veb'$ increases the mass coordinate by $\varepsilon$ and decreases the $j$-th margin coordinate by $\varepsilon^{j+1}$.

	We apply the perturbation results of \citet{MegiddoChandrasekaran1989perturbation} to the row-scaled system
	\[
		\Sigma A\vev=\Sigma\veb.
	\]
	The matrix $\Sigma A$ is integral, has rank $m$, and has the same entry bound $M$ as $A$.
	Our choice satisfies
	\[
		0<\varepsilon<\big((m!)^2M^{2m-1}\big)^{-1}.
	\]
	Their Proposition~3.1 and Corollary~3.2(ii) therefore imply the following two facts:
	\begin{enumerate}
		\item $A_B^{-1}\veb'$ has no zero coordinate for every basis $B$; and
		\item every basis feasible for $\veb'$ is feasible for $\veb$.
	\end{enumerate}
	Indeed, the all-positive perturbation $\Sigma\veb+\veeps$ of the row-scaled system corresponds to $\veb+\Sigma\veeps=\veb'$ in the original system.
	The bound on $\varepsilon$ depends on $m$ and $M$, but not on the number of columns of $A$.

	It remains to verify that $P(\veb')$ is feasible.
	Choose $\hat\rho\in\mathcal{Q}$ such that $d_{\hat\rho,c}\ge0$ for every $c\ne c^*$.
	For Borda, choose a ranking with $c^*$ first.
	For $k$-Approval, choose a ranking in $\mathcal{T}^{\sigma_1}$ whose top-$k$ set contains $c^*$.
	Starting from any feasible $(\vex,\ves)$ for $\veb$, increase $x_{\hat\rho}$ by $\varepsilon$ and set
	\[
		s'_{c_j}:=s_{c_j}+\varepsilon d_{\hat\rho,c_j}+\varepsilon^{j+1}.
	\]
	All new variables are nonnegative, and direct substitution gives the right-hand side $\veb'$.

	We next choose a finite-difference step before knowing an optimal basis of $P(\veb')$.
	For every basis $B$, Cramer's rule gives
	\[
		\left|(A_B^{-1})_{ij}\right|\le
		H:=(m-1)!M^{m-1}.
	\]
	Moreover, $2^{qm}\veb'$ is integral and every coordinate of $A_B^{-1}\veb'$ is nonzero.
	Another application of Cramer's rule gives the uniform lower bound
	\[
		\left|(A_B^{-1}\veb')_j\right|\ge
		\lambda:=\big(2^{qm}m!M^m\big)^{-1}.
	\]
	Both numbers have encoding length $\poly(L)$.
	Set
	\[
		\delta:=\frac{\lambda}{2H}.
	\]

	By standard LP theory, $P(\veb')$ has an optimal basis; fix one and call it $B$.
	Its basic variables are positive by nondegeneracy, and hence are at least $\lambda$.
	For every unit vector $\vece_i\in\mathbb{R}^m$,
	\[
		A_B^{-1}(\veb'+\delta\vece_i)
		\ge \frac{\lambda}{2}\veone>\vezero.
	\]
	Thus $B$ is feasible at $\veb'+\delta\vece_i$.
	It is also feasible at the original $\veb$ by the perturbation result above.

	Let
	\[
		\vey^*:=(A_B^\top)^{-1}\vep_B.
	\]
	By standard right-hand-side sensitivity analysis, an optimal basis remains optimal whenever it remains feasible, because its reduced costs do not change.
	The associated dual vector is unchanged and gives the exact change in the optimal value~\cite[Section~4.4; see also Theorem~5.2]{BertsimasTsitsiklis1997}.
	Consequently, $\vey^*$ is optimal for the dual objective $\veb^\top\vey$, and for every $i\in[m]$,
	\[
		V(\veb'+\delta\vece_i)=V(\veb')+\delta y_i^*.
	\]
	The $m+1$ value queries at $\veb'$ and $\veb'+\delta\vece_i$, $i\in[m]$, therefore recover
	\[
		y_i^*=
		\frac{V(\veb'+\delta\vece_i)-V(\veb')}{\delta}
		\qquad (i\in[m]).
	\]
	All queried right-hand sides are feasible and have encoding length $\poly(L)$.
	Return $\vey^*$.
\end{proof}

\begin{proof}[Proof of Lemma~\ref{lem:bribery-to-separation}]
	Proposition~\ref{prop:achievable-margins} and Lemma~\ref{lem:value-to-dual} give an optimization oracle for $K$ on every rational objective.
	Each value query is implemented with polynomially many calls to the decision algorithm $\mathcal{A}$.
	All calls keep the data $D$ fixed; only the right-hand side and budget change.
	Their encoding lengths are polynomial in $L$.

	Since $K$ is well-described, Proposition~\ref{prop:gls-equivalence} converts this optimization oracle into a separation oracle using polynomially many oracle calls and polynomial additional time.
	If $\mathcal{A}$ runs in time $f(k)\poly(L)$, every call has the same value of $k$.
	The stated running-time claims follow.
\end{proof}

\subsection{Hardness for Borda and $k$-Approval}\label{sec:borda-kapp-hardness}

\paragraph{Borda}
With general costs, Borda-\textsc{Swap Bribery}$_\infty$ admits no polynomial-time algorithm unless $P=NP$\footnote{\label{ftnt:turing}Formally, NP-hardness under Turing reductions: a polynomial-time bribery algorithm would still yield $P = NP$.} -- even if only a single voter type can be bribed.
We prove this via the margin version and the dual.
For a single bribable type $(\sigma_1,\pi)$, the margin LP is the primal of Section~\ref{sec:oracle-to-dual} with $\mathcal{R}$ the set of all rankings and $d_{\rho,c} := \textrm{rank}(c,\rho) - \textrm{rank}(c^*,\rho)$, the Borda margin contributed by a unit of mass at $\rho$.
Separation over its dual polyhedron $K$ then encodes the unit-processing-time scheduling problem $1|\text{prec},\, p_j = 1|\sum_j w_j C_j$~\cite{Lawler1978sequencing,DBLP:journals/ior/LenstraK78}:

\begin{lemma} \label{lem:separation-np-hard}
	Separating over $K$ is NP-hard.
\end{lemma}

The bridge of Lemma~\ref{lem:bribery-to-separation} carries this to the bribery problem, and a padding argument with unbribable voter types then removes the margins.

\begin{theorem} \label{thm:swap-bribery-hardness}
	Unless $P=NP$, there is no polynomial-time algorithm for Borda Margin-\textsc{Swap Bribery}$_\infty$, even with a single voter type, under polynomial-time Turing reductions.
\end{theorem}

\begin{corollary} \label{cor:margin-to-borda}
	Unless $P=NP$, there is no polynomial-time algorithm for Borda-\textsc{Swap Bribery}$_\infty$.
\end{corollary}

Note that $|d_{\rho,c}| \le m-1$, and that separating over $K$ means, given $(\alpha, \vezeta)$, either finding a $c$ with $\zeta_c < 0$, or finding a ranking $\rho \in \mathcal{R}$ with $\alpha + \sum_{c \neq c^*} d_{\rho,c}\, \zeta_{c} > \pi(\rho)$, or asserting that neither exists.

The reduction behind Lemma~\ref{lem:separation-np-hard} pins $c^*$ last. Then $d_{\rho,c_j}$ is the position of $c_j$ in $\rho$ up to a constant, so a dual constraint of $K$ is a weighted sum of the opponents' positions. Separating over $K$ \emph{is} therefore the NP-hard problem of finding a linear extension of a partial order that minimizes $\sum_j w_j \cdot (\text{position of } j)$~\cite{Lawler1978sequencing,DBLP:journals/ior/LenstraK78}.
	Two things do not line up. Separation maximizes where scheduling minimizes, and the dual insists on $\vezeta \ge \vezero$, so the weights cannot simply be negated. Instead they enter complemented, $\zeta_{c_j} := W - w_j$ for $W$ above every $w_j$, which flips the sense. The pinning and complementing operations both induce constants in the constraints, but those are identical for every ordering. We can thus handle them using $\alpha$, the one coordinate the dual leaves free in sign. We also use $\alpha$ to encode the scheduling threshold $\theta$.
	Finally, a bribery instance has no notion of precedence, so we make the forbidden swaps prohibitively expensive; the constraint of a ranking that is not a legal schedule is then satisfied outright, because its cost alone exceeds everything the profit can supply.
	Together, a constraint is violated exactly when its ranking is a legal schedule that comes in under $\theta$.

\begin{proof}[Proof of Lemma~\ref{lem:separation-np-hard}]
	We reduce from the unit-processing-time scheduling problem $1|\text{prec},\, p_j = 1|\sum_j w_j C_j$, equivalently \textsc{Minimum Weighted Linear Extension}, which is NP-hard~\cite{Lawler1978sequencing,DBLP:journals/ior/LenstraK78}:
	\niceproblem{($1|\text{prec},\, p_j = 1|\sum_j w_j C_j$)}
	{A set of jobs $J=\{1,\dots,n\}$, a partial order $\prec$ on $J$, weights $\vew \in \Z_{\geq 0}^n$, and a threshold value $\theta \in \Z$.}
	{A linear extension $\sigma$ of $\prec$ such that $\sum_{j\in J} w_j\,\mathrm{pos}_\sigma(j) \ < \theta$.}
	(The conventional decision threshold $\sum_{j \in J} w_j\,\mathrm{pos}_\sigma(j) \le B$ is captured by $\theta := B+1$, as the objective is integral.)

	Starting from an instance of ($1|\text{prec},\, p_j = 1|\sum_j w_j C_j$), we define an instance of the separation problem as follows.
	First, introduce a candidate $c_j$ for each $j \in J$, and an extra candidate $c^*$. Therefore, $C=\{c_j\mid j \in J\}\cup \{c^*\}$.
	Let $\rho^*$ be any ranking of $C$ such that $c^*$ is ranked last in $\rho^*$ and, for any $i\prec j$, $c_i \pref_{\rho^*} c_j$.
	The single bribable type starts from $\rho^*$, so every cost $\pi(\rho)$ below is the swap cost of reaching $\rho$ from $\rho^*$.

	For each $j\in J$, define $\zeta_{c_j}= W - w_j$, where $W=\max_{j \in J} \{w_j+1\}$.
	This way we guarantee that $\zeta_{c_j}> 0$, for all $c_j \neq c^*$.
	Set $\alpha = \theta - W \sum_{i=1}^n i + (n+1) \sum_{j \in J }(W-w_j)$, and let $H := 1 + |\alpha| + (m-1)\sum_{c \neq c^*} \zeta_c$ be a \emph{prohibitive} cost: since $|d_{\rho,c}| \le m-1$ and $\zeta_c > 0$, every $\rho \in \mathcal{R}$ satisfies $\alpha + \sum_{c \neq c^*} d_{\rho,c} \zeta_c \le |\alpha| + (m-1)\sum_{c \neq c^*} \zeta_c = H - 1 < H$.
	Define the swap cost function $\delta$ as follows:
	\begin{itemize}
		\item for every $j \in [n]$, $\delta(c_j,c^*) = H$,
		\item for every $i, j \in [n]$ with $i \prec j$, set $\delta(c_i,c_j) = H$,
		\item otherwise $\delta(c_i, c_j) = 0$.
	\end{itemize}
	The above costs enforce that for every $\rho \in \mathcal{R}$, $\pi(\rho) \ge H$ if either $c^*$ has been moved from the last position, or if $\rho$ violates the precedence constraints (any such $\rho$ inverts at least one pair of cost $H$); otherwise $\pi(\rho) = 0$.
	In particular, all costs are finite rationals, as the framework of Section~\ref{sec:oracle-to-dual} requires of the bribable type.

	We will need the following claim:
	\begin{claim} \label{claim:inequility:NP_hard:separation}
		Let $\sigma$ be an ordering of $J$ and $\rho\in \mathcal{R}$ a ranking of $C$ such that:
		\begin{itemize}
			\item $\textrm{rank}(c^*,\rho)=n+1$, and
			\item $c_j\pref_{\rho} c_i$ if and only if $j\prec_{\sigma} i$.
		\end{itemize}
		Then we have that:
		\begin{align*}
			\sum_{c \neq c^*} d_{\rho,c} \zeta_c  =  \theta -\alpha - \sum_{j \in J } \textrm{pos}_{\sigma} (j) w_j
		\end{align*}
	\end{claim}

	\begin{proofclaim}
		Compute $\sum_{c \neq c^*} d_{\rho,c} \zeta_c$ as follows:

		\begin{align*}
			 & \sum_{c \neq c^*} d_{\rho,c} \zeta_c  =                                                                                         \\
			 & = \sum_{c_j \neq c^*} \Big(\textrm{rank}(c_j,\rho) - n-1 \Big) (W-w_j)                                                           \\
			 & = \sum_{j \in J } \textrm{pos}_{\sigma} (j) (W-w_j) - (n+1) \sum_{j \in J } (W-w_j)                                              \\
			 & = - \sum_{j \in J } \textrm{pos}_{\sigma} (j) w_j + W \sum_{j \in J } \textrm{pos}_{\sigma} (j) \\
			 & \qquad - (n+1) \sum_{j \in J } (W-w_j) \\
			 & = - \sum_{j \in J } \textrm{pos}_{\sigma} (j) w_j + W \sum_{i=1}^n i - (n+1) \sum_{j \in J }(W-w_j)
		\end{align*}
		Finally, since, $\alpha = \theta - W \sum_{i=1}^n i + (n+1) \sum_{j \in J }(W-w_j)$:
		\begin{align*}
			\sum_{c \neq c^*} d_{\rho,c} \zeta_c  =  \theta -\alpha - \sum_{j \in J } \textrm{pos}_{\sigma} (j) w_j
		\end{align*}
	\end{proofclaim}

	Now, we will prove that there is an ordering of $J$ that extends $\prec$ and respects the threshold $\theta$ if and only if there is an ordering $\rho$ of $C$ such that \[\alpha + \sum_{c \neq c^*} d_{\rho,c}\zeta_c- \pi(\rho) > 0 \enspace .\]

	We first prove the following claim:
	\begin{claim}
			Let $\sigma$ be an ordering of $J$ that extends $\prec$.
			Suppose that $\sum_{j\in J} w_j \mathrm{pos}_{\sigma}(j) < \theta$.
			Let $\rho \in \mathcal{R}$ be a ranking over $C$ such that:
		\begin{itemize}
			\item $\textrm{rank}(c^*,\rho)=n+1$, and
			\item $c_j\pref_{\rho} c_i$ if and only if $j\prec_{\sigma} i$.
		\end{itemize}
		Then, we have that: $$\alpha +\sum_{c \neq c^*}d_{\rho,c} \zeta_c - \pi(\rho) >0$$
	\end{claim}

	\begin{proofclaim}
		By the definition of $\rho$, we have that $\pi(\rho)=0$. Thus, we need to show that:
		\[\alpha +\sum_{c \neq c^*} d_{\rho,c} \zeta_c >0 \enspace .\]
		Then, by Claim~\ref{claim:inequility:NP_hard:separation}:
		\[
			\sum_{c \neq c^*} d_{\rho,c} \zeta_c  =  \theta -\alpha - \sum_{j \in J } \textrm{pos}_{\sigma} (j) w_j > -\alpha,
		\]
		as desired.
	\end{proofclaim}

	For the reverse direction, we will prove the following:
	\begin{claim}
		Let $\rho \in \mathcal{R}$ be a ranking of $C$ such that \[
			\alpha +\sum_{c \neq c^*}d_{\rho,c} \zeta_c - \pi(\rho) >0 \enspace .\]
		Let $\sigma$ be the ranking of $J$ where $j\prec_{\sigma} i$ iff $c_j\pref_{\rho} c_i$. Then, $\sigma$ extends $\prec$ and
		$\sum_{j\in J} w_j \mathrm{pos}_{\sigma}(j) < \theta$.
	\end{claim}

	\begin{proofclaim}
		Since \[\alpha +\sum_{c \neq c^*}d_{\rho,c} \zeta_c - \pi(\rho) >0\] we have that $\pi(\rho) = 0$: otherwise $\pi(\rho) \ge H$, while $\alpha + \sum_{c \neq c^*} d_{\rho,c} \zeta_c \le H - 1$, so the inequality would not hold.
		Therefore, by construction of $\pi$, $\textrm{rank}(c^*,\rho) = n+1$.
		For the same reason, we can conclude that $\sigma$ must respect $\prec$ as otherwise $\pi(\rho) \ge H$.
		Now, we can apply Claim~\ref{claim:inequility:NP_hard:separation}, and thus:
		\begin{align*}
			\alpha +\sum_{c \neq c^*}d_{\rho,c} \zeta_c - \pi(\rho) >0                      & \implies \\
			\alpha + \Big(\theta-\alpha -\sum_{j \in J } \textrm{pos}_{\sigma} (j) w_j\Big)- 0 >0 & \implies \\
			\theta >\sum_{j \in J } \textrm{pos}_{\sigma} (j) w_j
		\end{align*}
		Since, $\sigma$ extends $\prec$ and respects the threshold $\theta$ it is indeed the desired ordering of $J$.
	\end{proofclaim}

	Put together, the claims above complete the reduction from ($1|\text{prec},\, p_j = 1|\sum_j w_j C_j$) to the separation problem over $K$.
\end{proof}

\begin{proof}[Proof of Theorem~\ref{thm:swap-bribery-hardness}]
	A polynomial-time algorithm for the problem in particular would decide budgets on single-type instances for which some finite-cost move achieves the margins, and would thus be the algorithm $\mathcal{A}$ of Lemma~\ref{lem:bribery-to-separation} in the Borda instantiation.
	It would therefore yield a polynomial-time separation algorithm over $K$, contradicting Lemma~\ref{lem:separation-np-hard}.
\end{proof}

	The proof below simulates the margins with unbribable ``padding'' voters. For each $c\neq c^*$, the pair of types associated with $c$ contributes $\Delta_c$ to $\mathrm{sc}_P(c)-\mathrm{sc}_P(c^*)$ and contributes zero to every other score difference relative to $c^*$.

\begin{proof}[Proof of Corollary~\ref{cor:margin-to-borda}]
	Let $\mathcal{I} = \big(C, \{(\sigma_i, \pi_i)\}_{i \in [\tau]}, \vemu, \veDelta\big)$ be an instance of \textrm{Borda Margin}-\textsc{Swap Bribery}$_\infty$.
	We construct an equivalent instance $\mathcal{I}'$ of Borda-\textsc{Swap Bribery}$_\infty$.
	First, copy the original voter types and their populations.
	Then, for each candidate $c\neq c^*$ with $\Delta_c \neq 0$, fix an ordering $c_1,\ldots,c_{m-2}$ of the candidates in $C\setminus \{c^*,c\}$ and create two new voter types, both of population $\frac{|\Delta_c|}{m}$:
	\begin{itemize}
		\item if $\Delta_c > 0$: $\rho^c_1$ is the ranking $c \prec c^* \prec c_1\prec\ldots \prec c_{m-2}$, and $\rho^c_2$ is the ranking $c\prec c_{m-2}\prec \ldots \prec c_1\prec c^*$;
		\item if $\Delta_c < 0$: $\rho^c_1$ is the ranking $c^*\prec c_1\prec\ldots \prec c_{m-2}\prec c$, and $\rho^c_2$ is the ranking $c_{m-2}\prec\ldots \prec c_1\prec c^* \prec c$.
	\end{itemize}
	All pairwise swap costs of the new types are $+\infty$; hence any ranking other than the type's initial one has cost $+\infty$, and a finite-cost move leaves the padding in place.
	The construction takes time polynomial in the encoding length of $\mathcal{I}$.

	Finite-cost moves in $\mathcal{I}'$ therefore coincide with finite-cost moves in $\mathcal{I}$ (extended by leaving the padding unchanged), and corresponding moves have equal costs; it remains to compare the winning conditions.
	Let $\mathrm{sc}_P$ denote the Borda score contributed by the padding types alone.
	We claim that
	\begin{equation} \label{eq:borda-padding-invariant}
		\mathrm{sc}_P(c) - \mathrm{sc}_P(c^*) = \Delta_c \qquad \text{for every } c \neq c^* .
	\end{equation}
	Consider the pair of types associated with $c$, per unit of population.
	If $\Delta_c > 0$, then $\rho^c_1$ gives $c$ exactly $1$ point more than $c^*$ and $\rho^c_2$ gives $c$ exactly $m-1$ points more, so the pair contributes $\frac{\Delta_c}{m}\big(1 + (m-1)\big) = \Delta_c$ to the difference $\mathrm{sc}_P(c) - \mathrm{sc}_P(c^*)$.
	If $\Delta_c < 0$, then symmetrically $\rho^c_1$ gives $c^*$ exactly $m-1$ points more than $c$ and $\rho^c_2$ gives $c^*$ exactly $1$ point more, so the pair contributes $\frac{-\Delta_c}{m} \cdot m = -\Delta_c$ to $\mathrm{sc}_P(c^*) - \mathrm{sc}_P(c)$ -- again matching~\eqref{eq:borda-padding-invariant}.
	Finally, the pair associated with any $c' \notin \{c, c^*\}$ contributes equally to $c$ and $c^*$: summed over its two rankings, every candidate outside $\{c'\}$ receives the same total score per unit of population ($m-2$ in the case $\Delta_{c'} > 0$ and $m$ in the case $\Delta_{c'} < 0$), as the second ranking reverses the block containing them.
	This proves~\eqref{eq:borda-padding-invariant}.

	Consequently, for any finite-cost move and every $c \neq c^*$, the score difference of $c^*$ over $c$ in $\mathcal{I}'$ equals its score difference in $\mathcal{I}$ minus $\Delta_c$.
	Hence, $c^*$ is a Borda winner in $\mathcal{I}'$ if and only if the corresponding move makes the score of $c^*$ at least $\Delta_c$ higher than that of $c$ for every $c \neq c^*$ in $\mathcal{I}$; that is, if and only if it makes $c^*$ a Borda Margin winner in $\mathcal{I}$.
	As corresponding moves have equal costs, the two instances are equivalent.
\end{proof}

\paragraph{$k$-Approval}

With general costs, $k$-Approval-\textsc{Swap Bribery}$_\infty$ is also intractable. To prove this, we are using the same template as for Borda-\textsc{Swap Bribery}$_\infty$. In this case, however, we show that the pricing problem is $W[1]$-hard when parameterized by $k$:

\begin{theorem} \label{thm:k-approval-separation-hardness}
    $k$-\textsc{Approval-Swap Pricing} is $W[1]$-hard parameterized by $k$, even when $\CandRank{\sigma}{c^*}=1$, and the pairwise swap costs are all $0$ or $1$.
\end{theorem}
\begin{proof}[Proof of Theorem~\ref{thm:k-approval-separation-hardness}]
    We present a parameterized reduction from \textsc{Clique}. Let $(G,r)$ be an instance of
    \textsc{Clique}, where $G=(V,E)$ and $r\geq 2$. By adding one isolated vertex if necessary,
    we may assume that $|V|>r$; this does not change whether $G$ has a clique of size $r$.
    Let $V=\{v_1,\ldots,v_n\}$.

    We construct an instance of $k$-\textsc{Approval-Swap Pricing} as follows.
    Create a distinguished candidate $c^*$ and, for each vertex $v_i$, a candidate $c_i$.
    Let the starting ranking be
    \[
        \sigma = c^* \prec c_1 \prec c_2 \prec \cdots \prec c_n,
    \]
    and set $k=r+1$.

    Next, we define the pairwise swap costs:
    \begin{itemize}
        \item $\delta(c^*,c_i)=0$, for each $i\in [n]$,
        \item $\delta(c_i,c_j)=1$, if $i<j$ and $v_i v_j\in E(G)$,
        \item $\delta(c_i,c_j)=0$, if $i<j$ and $v_i v_j\notin E(G)$.
    \end{itemize}

    For each $i\in [n]$, define
    $P_i=\sum_{j<i}\delta(c_j,c_i)$,
    and let
    $B=1+\sum_{i=1}^n P_i$.
    Since $B>P_i$ for every $i\in [n]$, the numbers
    $\zeta_{c_i}=B-P_i$
    are all positive.

    For every set $U\subseteq V(G)$ of size exactly $r$, let $S_U$ be the set of candidates $\{c^*\}\cup \{c_i\mid v_i\in U\}$,
    and let $\rho_U := \hat\rho_{S_U,\sigma}$ denote the stable ranking with top-$k$ set $S_U$.

    We first show that it is safe to restrict our attention to rankings whose top-$k$ set contains $c^*$. Additionally, every such ranking in $\mathcal{T}^\sigma$ is $\rho_U$ for some $U$.

    \begin{claim} \label{claim:k:Approval:Separation:hardness-p:ranked:first}
        Every optimal ranking approves $c^*$, and hence equals $\rho_U$ for a set $U \subseteq V(G)$ of size $r$.
    \end{claim}

    \begin{proofclaim}
        Consider any ranking $\rho'\in \mathcal{T}^\sigma$. If $I_{c^*,\rho'}=0$, then
        \begin{align*}
            & -\pi(\rho')+\sum_{c \neq c^*} (I_{c^*,\rho'}-I_{c,\rho'})\zeta_c
            \\ &  =-\pi(\rho')-\sum_{c \neq c^*} I_{c,\rho'}\zeta_c \\ & \leq 0.
        \end{align*}
        On the other hand, for the starting ranking $\sigma$ we have $\pi(\sigma)=0$ and
        \[
            -\pi(\sigma)+\sum_{c \neq c^*} (I_{c^*,\sigma}-I_{c,\sigma})\zeta_c
            =\sum_{i=k}^{n} \zeta_{c_i} > 0,
        \]
        because $n>r$ and every $\zeta_{c_i}$ is positive. Hence every optimal ranking approves $c^*$.

        Since every ranking in $\mathcal{T}^\sigma$ preserves the starting order within its top-$k$ block, and $c^*$ is ranked first in $\sigma$, every ranking in $\mathcal{T}^\sigma$ that approves $c^*$ also ranks $c^*$ first. Therefore, such a ranking is uniquely determined by its remaining $r$ approved candidates; that is, it is $\rho_U$ for some $U \subseteq V(G)$ with $|U|=r$.
    \end{proofclaim}

    So it suffices to compare the rankings $\rho_U$, where $U\subseteq V(G)$ and $|U|=r$.
    We now compute the cost of $\rho_U$.

    \begin{claim} \label{claim:k:Approval:Separation:hardness-relation:cost:set}
        For every $U\subseteq V(G)$ with $|U|=r$, we have
        \[
            \pi(\rho_U)= \sum_{v_i\in U} P_i - \sum_{\{v_i,v_j\}\subseteq U}\delta(c_i,c_j).
        \]
    \end{claim}

    \begin{proofclaim}
        In the ranking $\rho_U$, the only inversions with respect to the starting ranking $\sigma$
        are pairs $(c_j,c_i)$ such that $j<i$, $v_j\notin U$, and $v_i\in U$.
        Thus
        \[
            \pi(\rho_U)=\sum_{v_i\in U}\sum_{\substack{j<i\\ v_j\notin U}}\delta(c_j,c_i).
        \]
        On the other hand, by the definition of $P_i$,
        \begin{align*}
            \sum_{v_i\in U} P_i
            = &
            \sum_{v_i\in U}\sum_{j<i}\delta(c_j,c_i)
            = \\
            = & \sum_{v_i\in U}\sum_{\substack{j<i\\ v_j\notin U}}\delta(c_j,c_i)
            +
            \sum_{\{v_i,v_j\}\subseteq U}\delta(c_i,c_j).
        \end{align*}

        Rearranging proves the claim.
    \end{proofclaim}

    We are now ready to evaluate the objective. Let
    $Y=\sum_{c \neq c^*} \zeta_c$.
    For any set $U\subseteq V(G)$ of size $r$, using
    Claim~\ref{claim:k:Approval:Separation:hardness-relation:cost:set}, we obtain
    \begin{align*}
        &-\pi(\rho_U)+\sum_{c \neq c^*} (I_{c^*,\rho_U}-I_{c,\rho_U})\zeta_c = \\
        &Y-\sum_{v_i\in U} \zeta_{c_i} - \pi(\rho_U) = \\
        &Y-\sum_{v_i\in U} (B-P_i) - \left(\sum_{v_i\in U} P_i - \sum_{\{v_i,v_j\}\subseteq U}\delta(c_i,c_j)\right)=\\
        &Y-rB+\sum_{\{v_i,v_j\}\subseteq U}\delta(c_i,c_j).
    \end{align*}
    Since $\delta(c_i,c_j)=1$ exactly when $v_i v_j\in E(G)$, the final term is precisely the
    number of edges induced by $U$. Therefore,
    \[
        -\pi(\rho_U)+\sum_{c \neq c^*} (I_{c^*,\rho_U}-I_{c,\rho_U})\zeta_c
        =
        Y-rB+|E(G[U])|.
    \]
    Hence, this value is equal to $Y-rB+\binom{r}{2}$ if $U$ is a clique of size $r$, and is
    strictly smaller otherwise.

    Consequently, $G$ has a clique of size $r$ if and only if the optimum value of the constructed
    $k$-\textsc{Approval-Swap Pricing} instance is at least
    $Y-rB+\binom{r}{2}$.
    This is a parameterized reduction with $k=r+1$, so
    $k$-\textsc{Approval-Swap Pricing} is $W[1]$-hard when parameterized by $k$.
\end{proof}

The bridge of Lemma~\ref{lem:bribery-to-separation} carries this hardness of pricing to the bribery problem itself, and a padding argument with unbribable voter types again removes the margins.
\begin{theorem} \label{k-approval:hardness}
    Unless $FPT = W[1]$, $k$-Approval-\textsc{Swap Bribery}$_\infty$ cannot be solved in time $f(k) \poly(L)$ for any computable $f$, where $L$ is the encoding length of the input.
\end{theorem}
The proof is a parameterized Turing reduction: an $f(k)\poly(L)$-time bribery algorithm would give an FPT algorithm for \textsc{Clique}.

\begin{proof}[Proof of Theorem~\ref{k-approval:hardness}]
    Suppose, for contradiction, that $k$-Approval-\textsc{Swap Bribery}$_\infty$, or its budget-decision version, is solvable in time $f(k)\poly(L)$.
    By Claim~\ref{claim:k-approval-margin-padding} below, which preserves $C$, $c^*$, $k$, and all budgets while increasing the encoding length only polynomially, this yields an algorithm $\mathcal{A}$ of the same running-time form for (the budget-decision version of) $k$-Approval Margin-\textsc{Swap Bribery}$_\infty$ with a single voter type.
    The margin LP of such an instance is exactly the primal LP of Section~\ref{sec:oracle-to-dual} in its $k$-Approval instantiation, so Lemma~\ref{lem:bribery-to-separation} yields an algorithm for the separation problem over $K$ running in time $f(k)\poly(L)$ -- every call shares the LP data, hence the parameter $k$.
    Finally, separation over $K$ decides the hard instances constructed in the proof of Theorem~\ref{thm:k-approval-separation-hardness}: there all objective values $-\pi(\rho) + \sum_{c \neq c^*} (I_{c^*,\rho} - I_{c,\rho}) \zeta_c$ are integers, and $G$ has a clique of size $r$ if and only if their maximum is at least the integer $\Theta := Y - rB + \binom{r}{2}$; querying separation at the dual point $\vey = (\alpha, \vezeta)$ given by $\alpha := -(\Theta - \frac{1}{2})$ and the $\vezeta$ constructed there returns a violated constraint if and only if some $\rho \in \mathcal{T}^{\sigma_1}$ has objective value at least $\Theta$.
    Since the reduction is polynomial-time with $k = r+1$ and every step preserves $k$, this decides \textsc{Clique} in time $f(r+1)\poly$, placing the $W[1]$-hard \textsc{Clique} problem in FPT.
    It remains to prove the following claim.
\begin{claim} \label{claim:k-approval-margin-padding}
    Let $1\le k<m$, and let $\mathcal{I}= \big( C, \{(\sigma_i,\pi_i)\}_{i \in [\tau]}, \vemu, \veDelta \big)$ be an instance of $k$-Approval Margin-\textsc{Swap Bribery}$_\infty$.
    We can construct, in time $\poly(m,\tau,L)$, an instance $\mathcal{I}'$ of $k$-Approval \textsc{Swap Bribery}$_\infty$ with the following property.
    For every budget $\beta \in \mathbb{Q}$, the instance $\mathcal{I}$ admits a move of cost at most $\beta$ making $c^*$ a $k$-Approval Margin winner if and only if $\mathcal{I}'$ admits a move of cost at most $\beta$ making $c^*$ a $k$-Approval winner.
\end{claim}

\begin{proofclaim}
We start with a copy of each voter type $i\in [\tau]$. Then, for each $c \neq c^*$ such that $\Delta_c > 0$, we create $m-1$ new voter types $(\rho^c_i, \pi^c_i)_{i \in [m-1]}$, each with population $\frac{\Delta_c}{m-k}$, as follows. First, we fix an ordering $\{c_1,\ldots,c_{m-2}\}$ of the candidates in $C \setminus \{c,c^*\}$. Then, for any $i < m-1$, we define $\rho^c_i$ to be the ranking
\[
c \prec c_i \prec \cdots \prec c_{m-2} \prec c^* \prec c_1 \prec \cdots \prec c_{i-1},
\]
and $\rho^c_{m-1}$ to be the ranking
\[
c \prec c^* \prec c_1 \prec \cdots \prec c_{m-2}.
\]
All pairwise swap costs of each new type are $+\infty$, so that $\pi^c_i(\rho^c_i) = 0$ and $\pi^c_i(\rho)=+\infty$ for every $\rho \neq \rho^c_i$: the padding cannot be bribed at finite cost.

Note that adding the types $(\rho^c_i,\pi^c_i)_{i \in [m-1]}$ gives $\frac{(m-1)\Delta_c}{m-k}$ points to candidate $c$ (since $c$ is ranked first in all the new types) and $\frac{(k-1)\Delta_c}{m-k}$ points to any other candidate (since each of them is ranked within the first $k$ positions in exactly $k-1$ types). Thus, adding these voters gives $c$ exactly $\Delta_c$ more points than $c^*$, while the point difference between $c^*$ and $c'$, for any $c' \in C \setminus \{c,c^*\}$, is unchanged.

Now, for each $c \neq c^*$ such that $\Delta_c < 0$, we create $m-1$ new voter types $(\rho^c_i, \pi^c_i)_{i \in [m-1]}$, each with population $\frac{|\Delta_c|}{k}$, as follows. First, we fix an ordering $\{c_1,\ldots,c_{m-2}\}$ of the candidates in $C \setminus \{c,c^*\}$. Then, for any $i < m-1$, we define $\rho^c_i$ to be the ranking
\[
c_i \prec \cdots \prec c_{m-2} \prec c^* \prec c_1 \prec \cdots \prec c_{i-1} \prec c,
\]
and $\rho^c_{m-1}$ to be the ranking
\[
c^* \prec c_1 \prec \cdots \prec c_{m-2} \prec c.
\]
Again, all pairwise swap costs of each new type are $+\infty$.

Adding this set of voters gives $0$ points to candidate $c$ (since $c$ is ranked last in all the new types) and $k|\Delta_c|/k=|\Delta_c|$ points to any other candidate (since each of them is ranked within the first $k$ positions in exactly $k$ types). Thus, adding these voters gives $c^*$ exactly $|\Delta_c|$ more points than $c$, while the point difference between $c^*$ and $c'$, for any $c' \in C \setminus \{c,c^*\}$, is unchanged.

After repeating the previous steps for all candidates, we end up with at most $(m-1)^2$ new types, each constructed in polynomial time; call the result $\mathcal{I}'$.
In summary, writing $\mathrm{sc}_P$ for the $k$-Approval score contributed by the padding types alone, the construction guarantees
\[
	\mathrm{sc}_P(c) - \mathrm{sc}_P(c^*) = \Delta_c \qquad \text{for every } c \neq c^* .
\]

Since the padding types cannot be moved at finite cost, the finite-cost moves in $\mathcal{I}'$ coincide with the finite-cost moves in $\mathcal{I}$ (extended by leaving the padding in place), and corresponding moves have equal costs.
For any such move and every $c \neq c^*$, the score difference of $c^*$ over $c$ in $\mathcal{I}'$ equals its score difference in $\mathcal{I}$ minus $\Delta_c$; hence $c^*$ is a $k$-Approval winner in $\mathcal{I}'$ if and only if the same move makes $c^*$ a $k$-Approval Margin winner in $\mathcal{I}$.
Together with the equality of costs, this proves the claim.
\end{proofclaim}
\end{proof}

\paragraph{Removing infinite swap costs}
Infinite swap costs are a convenience, not an extension of the model.
In the discrete setting they are removed by charging $\beta+1$ per forbidden swap, which no bribery within budget $\beta$ can afford.
That argument fails here, because a move may shift an arbitrarily small amount of population: mass $\epsilon$ off an unbribable type costs only $\epsilon M$.

The replacement uses the following standard LP argument.
Consider the polytope of moves satisfying the mass-conservation and margin constraints, without imposing a budget constraint.
This polytope does not depend on the move costs.
If the instance obtained by replacing $+\infty$ with a finite cost has a solution within budget, then minimizing its cost over this polytope yields a vertex solution within budget.
Standard vertex-encoding bounds give a positive lower bound on every nonzero coordinate of such a vertex, depending only on the original instance.
We choose the replacement cost so that moving even this much mass to a formerly forbidden destination exceeds the budget.

\begin{lemma} \label{lem:remove-infinite-costs}
	Let $\mathcal{I} = \big(C, c^*, \ves, \{(\sigma_i,\delta_i)\}_{i \in [\tau]}, \vemu, \veDelta\big)$
	be an instance of $\mathcal{X}$~Margin-\textsc{Swap Bribery}$_\infty$ for a positional scoring rule
	$\mathcal{X}$ with score vector $\ves = (s_1 \ge \dots \ge s_m) \in \mathbb{Q}^m$, let
	$\beta \in \mathbb{Q}_{\ge 0}$ be a budget, and let $N$ be the encoding length of $(\mathcal{I}, \beta)$.
	For $M \in \mathbb{Q}_{>0}$ let $\mathcal{I}_M$ be the instance obtained from $\mathcal{I}$ by replacing
	every $+\infty$ entry of every $\delta_i$ by $M$ (and recomputing the $\pi_i$).
	Then
	\[
		M := 1 + \beta\, 2^{10N^2}
	\]
	is computable in time $\poly(N)$, has encoding length $\Oh(N^2)$, and
	$\mathcal{I}$ admits a move of cost at most $\beta$ whose output society achieves the margins
	$\veDelta$ if and only if $\mathcal{I}_M$ does.
\end{lemma}

\begin{proof}[Proof of Lemma~\ref{lem:remove-infinite-costs}]
	Write $\pi_i$ and $\pi^M_i$ for the move cost functions of $\mathcal{I}$ and $\mathcal{I}_M$,
	$\mathcal{A}_i := \{\rho \in \mathcal{R} : \pi_i(\rho) < +\infty\}$ for the \emph{allowed}
	destinations of type $i$ and $\mathcal{F}_i := \mathcal{R} \setminus \mathcal{A}_i$ for the
	forbidden ones, and let
	$d_{\rho,c} := s_{\rank(c^*,\rho)} - s_{\rank(c,\rho)}$.
	A move satisfies the required margins precisely when it belongs to the polytope
	\[
		P:=\left\{\vex\ge\vezero:
		\begin{array}{ll}
			\displaystyle\sum_{\rho\in\mathcal R}x_{i\to\rho}=\mu_i
			& (i\in[\tau]),\\[2mm]
			\displaystyle\sum_{i\in[\tau]}\sum_{\rho\in\mathcal R}
			d_{\rho,c}x_{i\to\rho}\ge\Delta_c
			& (c\ne c^*)
		\end{array}\right\}.
	\]
	The polytope $P$ is independent of $M$ and is bounded, because its nonnegative variables satisfy
	$\sum_\rho x_{i\to\rho}=\mu_i$ for every type $i$.

	We first record a standard bound on its vertices.
	There are $\tau+m-1$ constraints besides nonnegativity.
	Since the input explicitly lists the $m$ entries of the score vector and the $\tau$ voter types,
	we have $m,\tau\le N$, and hence $\tau+m-1\le 2N$.
	Therefore every vertex of $P$ has at most
	$2N$ positive coordinates.
	Restricting to these coordinates therefore gives a rational system of dimension at most $2N$.
	After clearing the denominators of the input rationals, the coefficients and right-hand sides of
	this system are integers of absolute value at most $2^{4N}$.
	The standard vertex-encoding bound for rational polyhedra \cite[Section~6.2]{GLS} therefore implies that
	for every vertex $\vex$ of $P$,
	\begin{equation}\label{eq:rm-inf-vertex-bound}
		x_{i\to\rho}=0\quad\text{or}\quad x_{i\to\rho}\ge 2^{-10N^2}
		\qquad(i\in[\tau],\ \rho\in\mathcal R).
	\end{equation}
	The constant is deliberately loose; only a bound of polynomial encoding length is needed.

	The number $M$ is obtained from $\beta$ and the $(10N^2+1)$-bit integer $2^{10N^2}$ by one
	multiplication and an addition. Thus it is computable in polynomial time, has encoding length
	$\Oh(N^2)$, and is positive.

	For the forward implication, a move $\vex$ of cost at most $\beta$ in $\mathcal{I}$ is supported
	on allowed columns, on which $\pi^M_i = \pi_i$.
	Thus $\vex$ is a move of $\mathcal{I}_M$ of the same cost with the same output society.

	Conversely, suppose that $\mathcal{I}_M$ admits a move in $P$ of cost at most $\beta$.
	Minimizing $\sum_{i\in[\tau]}\sum_{\rho\in\mathcal R}\pi^M_i(\rho)x_{i\to\rho}$ over the nonempty bounded polytope $P$
	yields an optimal vertex $\vex^*$ of cost at most $\beta$.
	If $\vex^*$ is supported only on allowed destinations, then it has the same cost in $\mathcal{I}$
	and is the required solution there.
	Otherwise, $x^*_{i\to\rho}>0$ for some $i$ and $\rho\in\mathcal F_i$.
	Such a ranking inverts at least one pair whose cost was $+\infty$ and is now $M$; since all swap
	costs are nonnegative, $\pi^M_i(\rho)\ge M$.
	Together with~\eqref{eq:rm-inf-vertex-bound}, this gives
	\[
		\sum_{j,\sigma}\pi^M_j(\sigma)x^*_{j\to\sigma}
		\ge Mx^*_{i\to\rho}
		\ge (1+\beta2^{10N^2})2^{-10N^2}>\beta,
	\]
	contrary to the choice of $\vex^*$.
\end{proof}

The magnitude of $M$ cannot be lowered to $\poly(N)$, and no single $M$ serves every budget.
Indeed, take $\ves = (1,0)$, one unbribable type $\sigma_1$ with $c \succ c^*$ of mass $1$, $\Delta_c = -1 + 2^{-L}$ and $\beta = 1$: then $\mathcal{I}$ is a no-instance while $\mathcal{I}_M$ is a yes-instance for every $M \le 2^{L+1}$.

\paragraph{The remaining case: uniform costs for Borda}

A natural and seemingly small generalization of unit costs are \emph{uniform} costs, where each type $i$ has a single per-swap cost $q_i > 0$.
For $k$-Approval they are already settled: uniform costs are additively separable, with $a^i_j := q_i$ and $b^i_{j'} := 0$, so Theorem~\ref{thm:k-approval-unit-separation} applies.
For Borda they are open, and the pricing problem of a type $i$ boils down to the following.
\niceproblem{\textsc{Inv-Perm-Max}}
{A vector $\vey \in \mathbb{Q}^n$.}
{A permutation $\rho$ of $[n]$ maximizing $\rho^T \vey - \mathrm{inv}(\rho)$, where $\mathrm{inv}(\rho) = |\{(i,j) \mid i < j,\ \rho(i) > \rho(j)\}|$.}

The correspondence is as follows.
Relabel the candidates in the order of $\sigma_i$ and set $n := m$, so that a ranking $\rho$ becomes a permutation of $[n]$ and $\pi_i(\rho) = q_i \cdot \mathrm{inv}(\rho)$.
Set $y_c := \zeta_c / q_i$ for every opponent $c$, and $y_{c^*} := -\big(\sum_{c \neq c^*} \zeta_c\big) / q_i$.
Then the pricing objective of type $i$ is $q_i \big(\rho^T \vey - \mathrm{inv}(\rho)\big)$.
Note that $\vey$ has one negative coordinate, which is why \textsc{Inv-Perm-Max} allows an arbitrary $\vey \in \mathbb{Q}^n$.

\begin{openproblem} \label{op:uniform}
	Is \textsc{Inv-Perm-Max} solvable in polynomial time?
\end{openproblem}

Despite the simplicity of its statement, \textsc{Inv-Perm-Max} resisted many proof attempts by both humans and AI.
At its core, one wants to place large values of $\rho$ to those indices where $\vey$ is also large, but these transfers have to be paid for, one unit per swap.
Various natural local improvement algorithms get stuck in local optima.
On the other hand, the structure is quite restrictive, and all our attempts to construct a hardness gadget have also failed.

\section{Future Work}\label{sec:future-work}

This case study is an invitation to a new scientific program.
The society continuum is not specific to bribery: it applies to any problem whose input is a population of discrete agents who can be meaningfully partitioned into types.
For each such problem one can ask the same questions. Does the continuum make it tractable? Does it embed its discrete version, as homogeneous rules like Kemeny and Slater do? Or is it hard for new, continuum-specific reasons?
A few directions strike us as particularly promising.

\paragraph{The rest of the bribery landscape}
First, how do the question marks in Table~\ref{tab:landscape} resolve?
For example, discrete Dodgson-\textsc{Bribery} lives \emph{above} NP, since already Dodgson winner determination is $\mathrm{P}^{\mathrm{NP}}_{||}$-complete~\cite{DBLP:journals/jacm/HemaspaandraHR97}; does its continuous variant drop to NP-completeness or below?
Next: candidate control, further cost models, and further rules await.
Copeland's rule is a test of our framework's limits: its winner regions are unions of exponentially many polyhedra, so neither Theorem~\ref{thm:poly-lp} nor the Configuration LP approach applies, at least not in any obvious way.

\paragraph{Discrete consequences}

For certain problems, LP/IP proximity results~\cite{DBLP:journals/mp/KnopKLMO23} apply and already decide the fate of all but $f(m,\tau)$ voters.
This yields both kernels~\cite{fomin_lokshtanov_saurabh_zehavi_2019} (effective preprocessing algorithms) and additive $\mathrm{OPT} + f(m,\tau)$ approximations, which may be quite useful when $n \gg m, \tau$.
When does this transfer work, when does it stop, and why?

\paragraph{Beyond voting}
Consider the society continuum model in the following subareas of computational social choice: matching under preferences, coalition formation, and fair division.
The discrete versions of problems there are defined over discrete agent populations, and studying the continuous variant is meaningful and, as far as we know, computationally wide open.

\section*{Disclosure on the Use of AI Tools}
We have used OpenAI Codex and Claude Code (versions ranging from January--July 2026) as writing assistants, for literature search, and as interactive assistants in developing, checking, and adversarially verifying the proofs.
All definitions, statements, and proofs were reviewed by the authors, who take full responsibility for the content.
\section*{Funding}
Martin Koutecký is partially supported by Charles University project UNCE 24/SCI/008, by the ERC-CZ project LL2406 of the Ministry of Education of Czech Republic, and by the project 25-17221S of GA ČR.
Nikolaos Melissinos is partially supported by Charles University projects UNCE 24/SCI/008 and PRIMUS 24/SCI/012, and by the project 25-17221S of GA ČR. 
Tung Anh Vu and Lluís Sabater are partially supported by the project 25-17221S of GA ČR.

\bibliographystyle{ACM-Reference-Format}
\bibliography{references}

@book{GreenGerber2019GOTV,
  author =        {Donald P. Green and Alan S. Gerber},
  edition =       {4},
  publisher =     {Brookings Institution Press},
  title =         {Get Out the Vote: How to Increase Voter Turnout},
  year =          {2019},
  isbn =          {9780815736943},
}

@book{Issenberg2012VictoryLab,
  author =        {Sasha Issenberg},
  publisher =     {Crown},
  title =         {The Victory Lab: The Secret Science of Winning
                   Campaigns},
  year =          {2012},
}

@misc{Cassidy2016Remain,
  author =        {John Cassidy},
  journal =       {The New Yorker},
  month =         jun,
  title =         {Why the Remain Campaign Lost the Brexit Vote},
  year =          {2016},
  url =           {https://www.newyorker.com/news/john-cassidy/why-the-remain-
                  campaign-lost-the-brexit-vote},
}

@book{pathria2016statistical,
  author =        {Pathria, Raj Kumar and Beale, Paul D.},
  edition =       {Third},
  publisher =     {Academic Press},
  title =         {Statistical Mechanics},
  year =          {2011},
}

@article{castellano2009statistical,
  author =        {Castellano, Claudio and Fortunato, Santo and
                   Loreto, Vittorio},
  journal =       {Reviews of modern physics},
  number =        {2},
  pages =         {591},
  publisher =     {APS},
  title =         {Statistical physics of social dynamics},
  volume =        {81},
  year =          {2009},
}

@book{galam2012sociophysics,
  author =        {Galam, S.},
  publisher =     {Springer New York},
  series =        {Understanding Complex Systems},
  title =         {Sociophysics: A Physicist's Modeling of
                   Psycho-political Phenomena},
  year =          {2012},
  isbn =          {9781461420316},
  url =           {https://books.google.cz/books?id=Tqm2P34aK-gC},
}

@book{pentland2014social,
  author =        {Pentland, Alex},
  publisher =     {Penguin},
  title =         {Social physics: How good ideas spread-the lessons
                   from a new science},
  year =          {2014},
}

@article{acemoglu2011opinion,
  author =        {Acemoglu, Daron and Ozdaglar, Asuman},
  journal =       {Dynamic Games and Applications},
  pages =         {3--49},
  publisher =     {Springer},
  title =         {Opinion dynamics and learning in social networks},
  volume =        {1},
  year =          {2011},
}

@article{gomes2010discrete,
  author =        {Gomes, Diogo A and Mohr, Joana and
                   Souza, Rafael Rigao},
  journal =       {Journal de math{\'e}matiques pures et appliqu{\'e}es},
  number =        {3},
  pages =         {308--328},
  publisher =     {Elsevier},
  title =         {Discrete time, finite state space mean field games},
  volume =        {93},
  year =          {2010},
}

@book{saari1994GeometryOfVoting,
  author =        {Donald G. Saari},
  publisher =     {Springer},
  title =         {Geometry of Voting},
  year =          {1994},
}

@article{hochbaumShamir1991HM,
  author =        {Dorit S. Hochbaum and Ron Shamir},
  journal =       {Operations Research},
  number =        {4},
  pages =         {648--653},
  title =         {Strongly Polynomial Algorithms for the High
                   Multiplicity Scheduling Problem},
  volume =        {39},
  year =          {1991},
  doi =           {10.1287/opre.39.4.648},
}

@article{fitzsimmonsHemaspaandra2019HM,
  author =        {Zack Fitzsimmons and Edith Hemaspaandra},
  journal =       {Autonomous Agents and Multi-Agent Systems},
  pages =         {383--402},
  title =         {High-multiplicity election problems},
  volume =        {33},
  year =          {2019},
  doi =           {10.1007/s10458-019-09410-4},
}

@inproceedings{elkind2009swapbribery,
  author =        {Edith Elkind and Piotr Faliszewski and
                   Arkadii Slinko},
  booktitle =     {Algorithmic Game Theory (SAGT 2009)},
  pages =         {299--310},
  publisher =     {Springer},
  series =        {Lecture Notes in Computer Science},
  title =         {Swap Bribery},
  volume =        {5814},
  year =          {2009},
  doi =           {10.1007/978-3-642-04645-2_27},
}

@article{DBLP:journals/iandc/BredereckCFNN16,
  author =        {Robert Bredereck and Jiehua Chen and
                   Piotr Faliszewski and Andr{\'{e}} Nichterlein and
                   Rolf Niedermeier},
  journal =       {Inf. Comput.},
  pages =         {140--164},
  title =         {Prices matter for the parameterized complexity of
                   shift bribery},
  volume =        {251},
  year =          {2016},
  bibsource =     {dblp computer science bibliography, https://dblp.org},
  doi =           {10.1016/j.ic.2016.08.003},
  url =           {https://doi.org/10.1016/j.ic.2016.08.003},
}

@inproceedings{hemaspaandraSchnoor2016dichotomy,
  author =        {Edith Hemaspaandra and Henning Schnoor},
  booktitle =     {{ECAI} 2016 -- 22nd European Conference on Artificial
                   Intelligence},
  pages =         {1071--1079},
  publisher =     {{IOS} Press},
  series =        {Frontiers in Artificial Intelligence and
                   Applications},
  title =         {Dichotomy for Pure Scoring Rules under Manipulative
                   Electoral Actions},
  volume =        {285},
  year =          {2016},
  doi =           {10.3233/978-1-61499-672-9-1071},
}

@inproceedings{BaumeisterHR19,
  author =        {Dorothea Baumeister and Tobias Hogrebe and Lisa Rey},
  booktitle =     {The Thirty-Third {AAAI} Conference on Artificial
                   Intelligence, {AAAI} 2019, The Thirty-First
                   Innovative Applications of Artificial Intelligence
                   Conference, {IAAI} 2019, The Ninth {AAAI} Symposium
                   on Educational Advances in Artificial Intelligence,
                   {EAAI} 2019, Honolulu, Hawaii, USA, January 27 -
                   February 1, 2019},
  pages =         {1764--1771},
  publisher =     {{AAAI} Press},
  title =         {Generalized Distance Bribery},
  year =          {2019},
  bibsource =     {dblp computer science bibliography, https://dblp.org},
  doi =           {10.1609/AAAI.V33I01.33011764},
  url =           {https://doi.org/10.1609/aaai.v33i01.33011764},
}

@article{dornSchlotter2012multivariate,
  author =        {Britta Dorn and Ildik{\'o} Schlotter},
  journal =       {Algorithmica},
  number =        {1},
  pages =         {126--151},
  title =         {Multivariate Complexity Analysis of Swap Bribery},
  volume =        {64},
  year =          {2012},
}

@article{elkindFaliszewskiSlinko2012HammingDR,
  author =        {Edith Elkind and Piotr Faliszewski and
                   Arkadii M. Slinko},
  journal =       {Social Choice and Welfare},
  number =        {4},
  pages =         {891--905},
  title =         {Rationalizations of Condorcet-consistent rules via
                   distances of {H}amming type},
  volume =        {39},
  year =          {2012},
  doi =           {10.1007/s00355-011-0555-0},
}

@article{faliszewskiKarpovObraztsova2022groupSeparable,
  author =        {Piotr Faliszewski and Alexander Karpov and
                   Svetlana Obraztsova},
  journal =       {Autonomous Agents and Multi-Agent Systems},
  number =        {1},
  pages =         {18},
  title =         {The complexity of election problems with
                   group-separable preferences},
  volume =        {36},
  year =          {2022},
  doi =           {10.1007/s10458-022-09549-7},
}

@article{young1977ExtendingCondorcet,
  author =        {H. Peyton Young},
  journal =       {Journal of Economic Theory},
  number =        {2},
  pages =         {335--353},
  title =         {Extending Condorcet's Rule},
  volume =        {16},
  year =          {1977},
}

@article{fishburn1977Condorcet,
  author =        {Peter C. Fishburn},
  journal =       {SIAM Journal on Applied Mathematics},
  number =        {3},
  pages =         {469--489},
  title =         {Condorcet Social Choice Functions},
  volume =        {33},
  year =          {1977},
}

@article{rotheSpakowskiVogel2003YoungExactComplexity,
  author =        {J{\"o}rg Rothe and Holger Spakowski and
                   J{\"o}rg Vogel},
  journal =       {Theory of Computing Systems},
  number =        {4},
  pages =         {375--386},
  title =         {Exact Complexity of the Winner Problem for {Y}oung
                   Elections},
  volume =        {36},
  year =          {2003},
}

@article{bartholdiToveyTrick1989Dodgson,
  author =        {Bartholdi, III, John J. and Tovey, Craig A. and
                   Trick, Michael A.},
  journal =       {Social Choice and Welfare},
  number =        {2},
  pages =         {157--165},
  title =         {Voting Schemes for Which It Can Be Difficult to Tell
                   Who Won the Election},
  volume =        {6},
  year =          {1989},
}

@article{DBLP:journals/jacm/HemaspaandraHR97,
  author =        {Edith Hemaspaandra and Lane A. Hemaspaandra and
                   J{\"o}rg Rothe},
  journal =       {J. {ACM}},
  month =         nov,
  number =        {6},
  pages =         {806--825},
  title =         {Exact Analysis of Dodgson Elections: Lewis Carroll's
                   1876 Voting System Is Complete for Parallel Access to
                   {NP}},
  volume =        {44},
  year =          {1997},
}

@article{schlotterFaliszewskiElkind2017campaign,
  author =        {Ildik{\'o} Schlotter and Piotr Faliszewski and
                   Edith Elkind},
  journal =       {Algorithmica},
  number =        {1},
  pages =         {84--115},
  title =         {Campaign Management under Approval-Driven Voting
                   Rules},
  volume =        {77},
  year =          {2017},
  doi =           {10.1007/s00453-015-0064-0},
}

@article{erdelyiFellowsRotheSchend2015bucklin,
  author =        {Gabor Erd{\'e}lyi and Michael R. Fellows and
                   J{\"o}rg Rothe and Lena Schend},
  journal =       {Journal of Computer and System Sciences},
  number =        {4},
  pages =         {632--660},
  title =         {Control complexity in {B}ucklin and fallback voting:
                   {A} theoretical analysis},
  volume =        {81},
  year =          {2015},
  doi =           {10.1016/j.jcss.2014.11.002},
}

@article{DBLP:journals/tcs/HemaspaandraSV05,
  author =        {Edith Hemaspaandra and Holger Spakowski and
                   J{\"{o}}rg Vogel},
  journal =       {Theor. Comput. Sci.},
  number =        {3},
  pages =         {382--391},
  title =         {The complexity of Kemeny elections},
  volume =        {349},
  year =          {2005},
  bibsource =     {dblp computer science bibliography, https://dblp.org},
  doi =           {10.1016/j.tcs.2005.08.031},
  url =           {https://doi.org/10.1016/j.tcs.2005.08.031},
}

@inproceedings{DBLP:conf/stacs/Lampis22,
  author =        {Michael Lampis},
  booktitle =     {39th International Symposium on Theoretical Aspects
                   of Computer Science, {STACS} 2022, March 15-18, 2022,
                   Marseille, France (Virtual Conference)},
  editor =        {Petra Berenbrink and Benjamin Monmege},
  pages =         {45:1--45:14},
  publisher =     {Schloss Dagstuhl - Leibniz-Zentrum f{\"{u}}r
                   Informatik},
  series =        {LIPIcs},
  title =         {Determining a Slater Winner Is Complete for Parallel
                   Access to {NP}},
  volume =        {219},
  year =          {2022},
  bibsource =     {dblp computer science bibliography, https://dblp.org},
  doi =           {10.4230/LIPIcs.STACS.2022.45},
  url =           {https://doi.org/10.4230/LIPIcs.STACS.2022.45},
}

@article{faliszewskiReischRotheSchend2015bucklin,
  author =        {Piotr Faliszewski and Yannick Reisch and
                   J{\"o}rg Rothe and Lena Schend},
  journal =       {Autonomous Agents and Multi-Agent Systems},
  number =        {6},
  pages =         {1091--1124},
  title =         {Complexity of Manipulation, Bribery, and Campaign
                   Management in {B}ucklin and Fallback Voting},
  volume =        {29},
  year =          {2015},
  doi =           {10.1007/s10458-014-9277-x},
}

@article{brandtBrillHemaspaandraHemaspaandra2015bypassing,
  author =        {Felix Brandt and Markus Brill and Edith Hemaspaandra and
                   Lane A. Hemaspaandra},
  journal =       {Journal of Artificial Intelligence Research},
  pages =         {439--496},
  title =         {Bypassing Combinatorial Protections: Polynomial-Time
                   Algorithms for Single-Peaked Electorates},
  volume =        {53},
  year =          {2015},
  doi =           {10.1613/jair.4647},
}

@inproceedings{fitzsimmonsHemaspaandraHooverNarvaez2019control,
  author =        {Zack Fitzsimmons and Edith Hemaspaandra and
                   Alexander Hoover and David E. Narv{\'a}ez},
  booktitle =     {Proceedings of the 33rd {AAAI} Conference on
                   Artificial Intelligence ({AAAI} 2019)},
  pages =         {1933--1940},
  title =         {Very Hard Electoral Control Problems},
  year =          {2019},
}

@article{betzlerDorn2010possibleWinner,
  author =        {Nadja Betzler and Britta Dorn},
  journal =       {Journal of Computer and System Sciences},
  number =        {8},
  pages =         {812--836},
  title =         {Towards a Dichotomy for the Possible Winner Problem
                   in Elections Based on Scoring Rules},
  volume =        {76},
  year =          {2010},
  doi =           {10.1016/j.jcss.2010.05.003},
}

@article{baumeisterRothe2012finalStep,
  author =        {Dorothea Baumeister and J{\"o}rg Rothe},
  journal =       {Information Processing Letters},
  number =        {5},
  pages =         {186--190},
  title =         {Taking the Final Step to a Full Dichotomy of the
                   Possible Winner Problem in Pure Scoring Rules},
  volume =        {112},
  year =          {2012},
  doi =           {10.1016/j.ipl.2011.11.016},
}

@book{saari1995BasicGeometry,
  author =        {Donald G. Saari},
  publisher =     {Springer},
  title =         {Basic Geometry of Voting},
  year =          {1995},
}

@book{saari2001DecisionsElections,
  author =        {Donald G. Saari},
  publisher =     {Cambridge University Press},
  title =         {Decisions and Elections: Explaining the Unexpected},
  year =          {2001},
}

@inproceedings{mccabeDansted2008DodgsonAbsurdity,
  author =        {John C. McCabe-Dansted},
  booktitle =     {Proceedings of the 2nd International Workshop on
                   Computational Social Choice (COMSOC 2008)},
  title =         {Dodgson's Rule Approximations and Absurdity},
  year =          {2008},
}

@incollection{meskanenNurmi2008DistRational,
  author =        {Tuukka Meskanen and Hannu Nurmi},
  booktitle =     {Power, Freedom, and Voting},
  editor =        {Matthew Braham and F. Steffen},
  publisher =     {Springer},
  title =         {Closeness Counts in Social Choice},
  year =          {2008},
}

@inproceedings{xiaConitzerProcaccia2010scheduling,
  author =        {Xia, Lirong and Conitzer, Vincent and
                   Procaccia, Ariel D.},
  booktitle =     {Proceedings of the 11th {ACM} Conference on
                   Electronic Commerce ({EC} '10)},
  pages =         {275--284},
  title =         {A Scheduling Approach to Coalitional Manipulation},
  year =          {2010},
}

@inproceedings{freemanBrillConitzer2015tiebreaking,
  author =        {Freeman, Rupert and Brill, Markus and
                   Conitzer, Vincent},
  booktitle =     {Proceedings of the 14th International Conference on
                   Autonomous Agents and Multiagent Systems ({AAMAS}
                   '15)},
  pages =         {1401--1409},
  title =         {General Tiebreaking Schemes for Computational Social
                   Choice},
  year =          {2015},
}

@inproceedings{meir2015plurality,
  author =        {Meir, Reshef},
  booktitle =     {Proceedings of the 29th {AAAI} Conference on
                   Artificial Intelligence ({AAAI} '15)},
  pages =         {2103--2109},
  title =         {Plurality Voting under Uncertainty},
  year =          {2015},
}

@article{GLS1981,
  author =        {Gr\"otschel, Martin and Lov\'asz, L\'aszl\'o and
                   Schrijver, Alexander},
  journal =       {Combinatorica},
  number =        {2},
  pages =         {169--197},
  title =         {The ellipsoid method and its consequences in
                   combinatorial optimization},
  volume =        {1},
  year =          {1981},
}

@inproceedings{JainMahdianSalavatipour2003packing,
  author =        {Jain, Kamal and Mahdian, Mohammad and
                   Salavatipour, Mohammad R.},
  booktitle =     {Proceedings of the 14th Annual {ACM-SIAM} Symposium
                   on Discrete Algorithms (SODA '03)},
  pages =         {266--274},
  title =         {Packing {S}teiner trees},
  year =          {2003},
}

@inproceedings{Xu2016security,
  author =        {Xu, Haifeng},
  booktitle =     {Proceedings of the 2016 {ACM} Conference on Economics
                   and Computation ({EC} '16)},
  pages =         {497--514},
  title =         {The Mysteries of Security Games: Equilibrium
                   Computation Becomes Combinatorial Algorithm Design},
  year =          {2016},
}

@inproceedings{bhaskarChengKoSwamy2016signaling,
  author =        {Bhaskar, Umang and Cheng, Yu and Ko, Young Kun and
                   Swamy, Chaitanya},
  booktitle =     {Proceedings of the 2016 {ACM} Conference on Economics
                   and Computation ({EC} '16)},
  title =         {Hardness Results for Signaling in {B}ayesian Zero-Sum
                   and Network Routing Games},
  year =          {2016},
}

@book{HandbookComSoC,
  editor =        {Felix Brandt and Vincent Conitzer and Ulle Endriss and
                   J{\'{e}}r{\^{o}}me Lang and Ariel D. Procaccia},
  publisher =     {Cambridge University Press},
  title =         {Handbook of Computational Social Choice},
  year =          {2016},
  bibsource =     {dblp computer science bibliography, https://dblp.org},
  doi =           {10.1017/CBO9781107446984},
  isbn =          {9781107446984},
  url =           {https://doi.org/10.1017/CBO9781107446984},
}

@inproceedings{DBLP:conf/atal/KnopKM18,
  author =        {Dušan Knop and Martin Kouteck{\'{y}} and
                   Matthias Mnich},
  booktitle =     {Proceedings of the 17th International Conference on
                   Autonomous Agents and MultiAgent Systems, {AAMAS}
                   2018, Stockholm, Sweden, July 10-15, 2018},
  editor =        {Elisabeth Andr{\'{e}} and Sven Koenig and
                   Mehdi Dastani and Gita Sukthankar},
  pages =         {256--264},
  publisher =     {International Foundation for Autonomous Agents and
                   Multiagent Systems Richland, SC, {USA} / {ACM}},
  title =         {A Unifying Framework for Manipulation Problems},
  year =          {2018},
  bibsource =     {dblp computer science bibliography, https://dblp.org},
  url =           {http://dl.acm.org/citation.cfm?id=3237427},
}

@article{McGarvey1953,
  author =        {David C. McGarvey},
  journal =       {Econometrica},
  number =        {4},
  pages =         {608--610},
  title =         {A Theorem on the Construction of Voting Paradoxes},
  volume =        {21},
  year =          {1953},
}

@inproceedings{XiaConitzer2008GSR,
  author =        {Lirong Xia and Vincent Conitzer},
  booktitle =     {Proceedings of the 9th {ACM} Conference on Electronic
                   Commerce ({EC})},
  pages =         {109--118},
  title =         {Generalized Scoring Rules and the Frequency of
                   Coalitional Manipulability},
  year =          {2008},
}

@article{MosselProcacciaRacz2013smooth,
  author =        {Elchanan Mossel and Ariel D. Procaccia and
                   Mikl{\'o}s Z. R{\'a}cz},
  journal =       {Journal of Artificial Intelligence Research},
  pages =         {923--951},
  title =         {A Smooth Transition from Powerlessness to Absolute
                   Power},
  volume =        {48},
  year =          {2013},
}

@inproceedings{XiaConitzer2009FLC,
  author =        {Lirong Xia and Vincent Conitzer},
  booktitle =     {Proceedings of the 21st International Joint
                   Conference on Artificial Intelligence ({IJCAI})},
  pages =         {336--341},
  title =         {Finite Local Consistency Characterizes Generalized
                   Scoring Rules},
  year =          {2009},
}

@book{GLS,
  author =        {Gr\"otschel, Martin and Lov\'asz, L\'aszl\'o and
                   Schrijver, Alexander},
  edition =       {Second},
  pages =         {xii+362},
  publisher =     {Springer-Verlag, Berlin},
  series =        {Algorithms and Combinatorics},
  title =         {Geometric algorithms and combinatorial optimization},
  volume =        {2},
  year =          {1993},
}

@inproceedings{elkindFaliszewski2010campaign,
  author =        {Edith Elkind and Piotr Faliszewski},
  booktitle =     {Internet and Network Economics - 6th International
                   Workshop, {WINE} 2010},
  pages =         {473--482},
  publisher =     {Springer},
  series =        {Lecture Notes in Computer Science},
  title =         {Approximation Algorithms for Campaign Management},
  volume =        {6484},
  year =          {2010},
  doi =           {10.1007/978-3-642-17572-5_40},
}

@book{HardyLittlewoodPolya1952,
  author =        {G. H. Hardy and J. E. Littlewood and G. P{\'o}lya},
  edition =       {2nd},
  publisher =     {Cambridge University Press},
  title =         {Inequalities},
  year =          {1952},
}

@article{MegiddoChandrasekaran1989perturbation,
  author =        {Nimrod Megiddo and R. Chandrasekaran},
  journal =       {Operations Research Letters},
  number =        {6},
  pages =         {305--308},
  title =         {On the $\varepsilon$-Perturbation Method for Avoiding
                   Degeneracy},
  volume =        {8},
  year =          {1989},
  doi =           {10.1016/0167-6377(89)90014-X},
}

@book{BertsimasTsitsiklis1997,
  address =       {Belmont, MA},
  author =        {Dimitris Bertsimas and John N. Tsitsiklis},
  publisher =     {Athena Scientific},
  title =         {Introduction to Linear Optimization},
  year =          {1997},
  isbn =          {978-1-886529-19-9},
}

@article{Lawler1978sequencing,
  author =        {Eugene L. Lawler},
  journal =       {Annals of Discrete Mathematics},
  pages =         {75--90},
  title =         {Sequencing Jobs to Minimize Total Weighted Completion
                   Time Subject to Precedence Constraints},
  volume =        {2},
  year =          {1978},
}

@article{DBLP:journals/ior/LenstraK78,
  author =        {Jan Karel Lenstra and A. H. G. Rinnooy Kan},
  journal =       {Oper. Res.},
  number =        {1},
  pages =         {22--35},
  title =         {Complexity of Scheduling under Precedence
                   Constraints},
  volume =        {26},
  year =          {1978},
  bibsource =     {dblp computer science bibliography, https://dblp.org},
  doi =           {10.1287/OPRE.26.1.22},
  url =           {https://doi.org/10.1287/opre.26.1.22},
}

@article{DBLP:journals/mp/KnopKLMO23,
  author =        {Dušan Knop and Martin Kouteck{\'{y}} and Asaf Levin and
                   Matthias Mnich and Shmuel Onn},
  journal =       {Math. Program.},
  number =        {1},
  pages =         {199--227},
  title =         {High-multiplicity N-fold {IP} via configuration {LP}},
  volume =        {200},
  year =          {2023},
  bibsource =     {dblp computer science bibliography, https://dblp.org},
  doi =           {10.1007/s10107-022-01882-9},
  url =           {https://doi.org/10.1007/s10107-022-01882-9},
}

@book{fomin_lokshtanov_saurabh_zehavi_2019,
  author =        {Fomin, Fedor V. and Lokshtanov, Daniel and
                   Saurabh, Saket and Zehavi, Meirav},
  publisher =     {Cambridge University Press},
  title =         {Kernelization: Theory of Parameterized Preprocessing},
  year =          {2019},
  doi =           {10.1017/9781107415157},
}

@inproceedings{DornSchlotter2010COMSOC,
  author    = {Britta Dorn and Ildik{\'o} Schlotter},
  title     = {Multivariate Complexity Analysis of Swap Bribery},
  booktitle = {Proceedings of the 3rd International Workshop on Computational Social Choice (COMSOC 2010)},
  year      = {2010},
  url       = {https://comsoc-community.org/assets/proceedings/comsoc-2010/Dorn.pdf}
}

\end{document}